\documentclass[twoside]{article}

\usepackage[accepted]{aistats2020}
\usepackage[round]{natbib}

\usepackage{hyperref}
\usepackage{float}
\usepackage{graphicx}
\usepackage[usenames,dvipsnames,svgnames]{xcolor}

\newsavebox{\mybox}

\usepackage{amssymb}
\usepackage{amsmath}
\usepackage{amsthm}
\usepackage{bm}
\usepackage{algorithm}
\usepackage{algorithmicx}
\usepackage{algpseudocode}

\DeclareMathSymbol{\mathdblquotechar}{\mathalpha}{letters}{`"}
\newcommand{\mathdblquote}{\mathtt{\mathdblquotechar}}
\begingroup\lccode`~=`"\lowercase{\endgroup
  \let~\mathdblquote
}
\AtBeginDocument{\mathcode`"="8000 }

\algrenewcomment[1]{\(\triangleright\) #1}
\algnewcommand{\LineComment}[1]{\State \(\triangleright\) #1}

\algblockx[exec]{execute}{END}%
    [1][Unknown]{Execute #1}%
    [1][Unknown]{\vspace{-4mm}}

\usepackage{listings}
\usepackage[framemethod=tikz]{mdframed}

\definecolor{addresscolor}{RGB}{22, 88, 12}
\definecolor{distcolor}{RGB}{40, 12, 88}
\definecolor{gencode}{RGB}{255,235,150}
\definecolor{jlcode}{RGB}{255,210,210}
\definecolor{dslcode}{RGB}{220,242,255}
\usepackage{colortbl}

\newtheorem{theorem}{Theorem}[section]

\newtheorem{lemma}[theorem]{Lemma}

\newcommand{\integers}{\mathbb{Z}}

\begin{document}

%

%

\twocolumn[

\aistatstitle{Automating Involutive MCMC using Probabilistic and Differentiable Programming}

\aistatsauthor{ Marco Cusumano-Towner \And Alexander K. Lew \And Vikash K. Mansinghka }

\aistatsaddress{ Massachusetts Institute of Technology } ]

\begin{abstract}
Involutive MCMC is a unifying mathematical construction for MCMC kernels that generalizes many classic and state-of-the-art MCMC algorithms, from reversible jump MCMC to kernels based on deep neural networks.
But as with MCMC samplers more generally, implementing involutive MCMC kernels is often tedious and error-prone, especially when sampling on complex state spaces.
This paper describes a technique for automating the implementation of involutive MCMC kernels given (i) a pair of probabilistic programs defining the target distribution and an auxiliary distribution respectively and (ii) a differentiable program that transforms the execution traces of these probabilistic programs.
The technique, which is implemented as part of the Gen probabilistic programming system, also automatically detects user errors in the specification of involutive MCMC kernels and exploits sparsity in the kernels for improved efficiency.
The paper shows example Gen code for a split-merge reversible jump move in an infinite Gaussian mixture model and a state-dependent mixture of proposals on a combinatorial space of covariance functions for a Gaussian process.
\end{abstract}

\section{INTRODUCTION}
Markov chain Monte Carlo (MCMC) algorithms are powerful tools for approximate sampling from probability distributions and are central to modern Bayesian statistics, probabilistic machine learning, statistical physics, and numerous application areas where probabilistic modeling and inference are used.
But designing and deriving efficient MCMC algorithms is mathematically involved, and implementing MCMC kernels is tedious and notoriously error-prone.
These challenges are especially pronounced when sampling from probability distributions on complex state spaces that combine symbolic, numeric, and structural uncertainty, such as those arising in computational biology~\citep{huelsenbeck2004bayesian}, robotics and scene understanding~\citep{geiger2011generative}, and models of human cognition~\citep{tenenbaum2011grow}.

Involutive MCMC is a mathematical construction for MCMC kernels that gives a simplifying and unifying perspective on a number of previously disparate classes of kernels,
including reversible jump MCMC~\citep{green1995reversible}, which is the dominant mathematical framework for MCMC on complex state spaces.
Involutive MCMC constructs an MCMC kernel from three components: (i) the unnormalized target density, (ii) a sampler and density for an auxiliary probability distribution, and (iii) an involution\footnote{An \emph{involution} is a bijection that is its own inverse (that is, $f$ where $f(f(z)) = z$).} on an extended state space.
While this construction is mathematically clarifying, correctly implementing involutive MCMC kernels on complex state spaces remains challenging due to tedious density and Jacobian computations and the need for careful reasoning about the state space.

This paper formulates involutive MCMC on general state spaces and shows how to automate the implementation of an involutive MCMC kernel from three declarative programs that define the target probability distribution, auxiliary probability distribution, and the involution, respectively.
The probability distributions are defined as probabilistic programs, and the involution is defined as a differentiable program that transforms the execution traces of the probabilistic programs.
We use probabilistic programming techniques and automatic differentiation to automatically compute the acceptance probability.
We also show how to automatically detect mathematical errors in the specification of an involutive MCMC kernel, and how to improve the efficiency by automatically exploiting the sparsity structure in the involution.
We implemented the approach within the Gen probabilistic programming system\footnote{\url{https://www.gen.dev}}~\citep{cusumano2019gen}.
The paper shows examples of involutive MCMC kernels implemented in Gen for (i) a split-merge reversible jump move in an infinite mixture model, and (ii) a state-dependent mixture of Metropolis-Hastings proposals on an infinite combinatorial space of covariance functions for a Gaussian process.
We also provide a lightweight PyTorch implementation of the basic approach at
{\small\url{https://github.com/probcomp/autoimcmc}}.

The contributions of this paper include:
\begin{enumerate}
\item A measure-theoretic formulation of involutive MCMC on general state spaces.
\item A mathematical formulation of state spaces consisting of arbitrary key-value stores (i.e. dictionaries) that formalizes the space of execution traces of Gen probabilistic programs, and a formulation of involutive MCMC on these spaces.
\item A differentiable programming language for defining transformations between spaces of traces of probabilistic programs.
\item An algorithm that automates the implementation of an involutive MCMC kernel given two probabilistic programs and a differentiable program encoding the involution, using automatic differentiation and tracing of probabilistic programs.
\item An extension to the algorithm that exploits sparsity in the involution to reduce the number of operations needed to apply an involutive MCMC kernel from $O(N^3)$ to $O(1)$ in some cases, where $N$ is the dimensionality of the latent space.
\item An algorithm that dynamically detects errors in the specification of an involutive MCMC kernel.
\end{enumerate}

\section{RELATED WORK}

The involutive MCMC construction was previously implemented~\citep{gen-inv-mcmc} by the authors as part of the inference library of the Gen probabilistic programming system~\citep{cusumano2019gen}.
The construction was motivated in part by a desire for a simple interface that automated the implementation of reversible jump MCMC samplers~\citep{green1995reversible}, state-dependent mixtures of proposals on complex state spaces, and data-driven neural proposals.
Gen's involutive MCMC construction has since been used by a number of researchers to design and implement MCMC algorithms in diverse domains, including computational biology~\citep{merrell2020inferring} and artificial intelligence~\citep{zhi2020online}.

\citet{imcmc} independently identified the involutive MCMC construction as a unifying framework for MCMC algorithms, and showed how more than a dozen classic and recent MCMC algorithms can be cast within this framework.
\citet{imcmc} also identified design principles for developing new MCMC algorithms using the involutive MCMC construction, and showed that the framework aids in the derivation of novel efficient MCMC algorithms.

The involutive MCMC construction encompasses many existing classes of MCMC kernels, some of which explicitly make use of bijective or involutive deterministic maps.
In particular, the reversible jump framework~\citep{green1995reversible,hastie2012model} employs a family of continuously differentiable bijections between the parameter spaces of different models.
\citet{tierney1998note} described a family of deterministic proposals based on a deterministic involution that is equivalent to involutive MCMC but without the auxiliary probability distribution.
More recently,~\citet{spanbauer2020deep} defined a class of deep generative models based on differentiable involutions and trained these models to serve as efficient proposal distributions on continuous state spaces; the resulting algorithm is an instance of the construction presented in this paper.

In recent decades, a number of probabilistic programming systems have automated the implementations of probabilistic inference algorithms~\citep{gilks1994language,milch2007,pfeffer200714,goodman2012church,gehr2016psi,carpenter2017stan}.
Most of these systems support generic built-in inference algorithms, with user-customization limited to tweaking algorithm parameters.
Some systems allow for user-defined variational families or proposal distributions~\citep{ritchie2016deep,bingham2019pyro}.
Programmable inference~\citep{mansinghka2018probabilistic} proposes that inference algorithms be programmed by users using new high-level inference abstractions.
The Gen probabilistic programming system~\citep{cusumano2019gen} exposes an API that supports high-level user implementations of an open-ended set of inference algorithms, and abstracts away low-level implementation details of inference algorithms.

One other probabilistic programming system besides Gen supports custom reversible jump samplers:
\citet{roberts2019} present a system embedded in Haskell that automatically generates the implementation of some reversible jump MCMC kernels from a high-level specification.
\citet{narayanan2020symbolic} give a technique that automatically computes Metropolis-Hastings acceptance probabilities in some settings; however, these approaches do not handle many kernels that can be handled by our technique, including the example in Figure~\ref{fig:structure-learning}.

\begin{figure*}[ht!]
    \centering
    \includegraphics[width=0.93\textwidth]{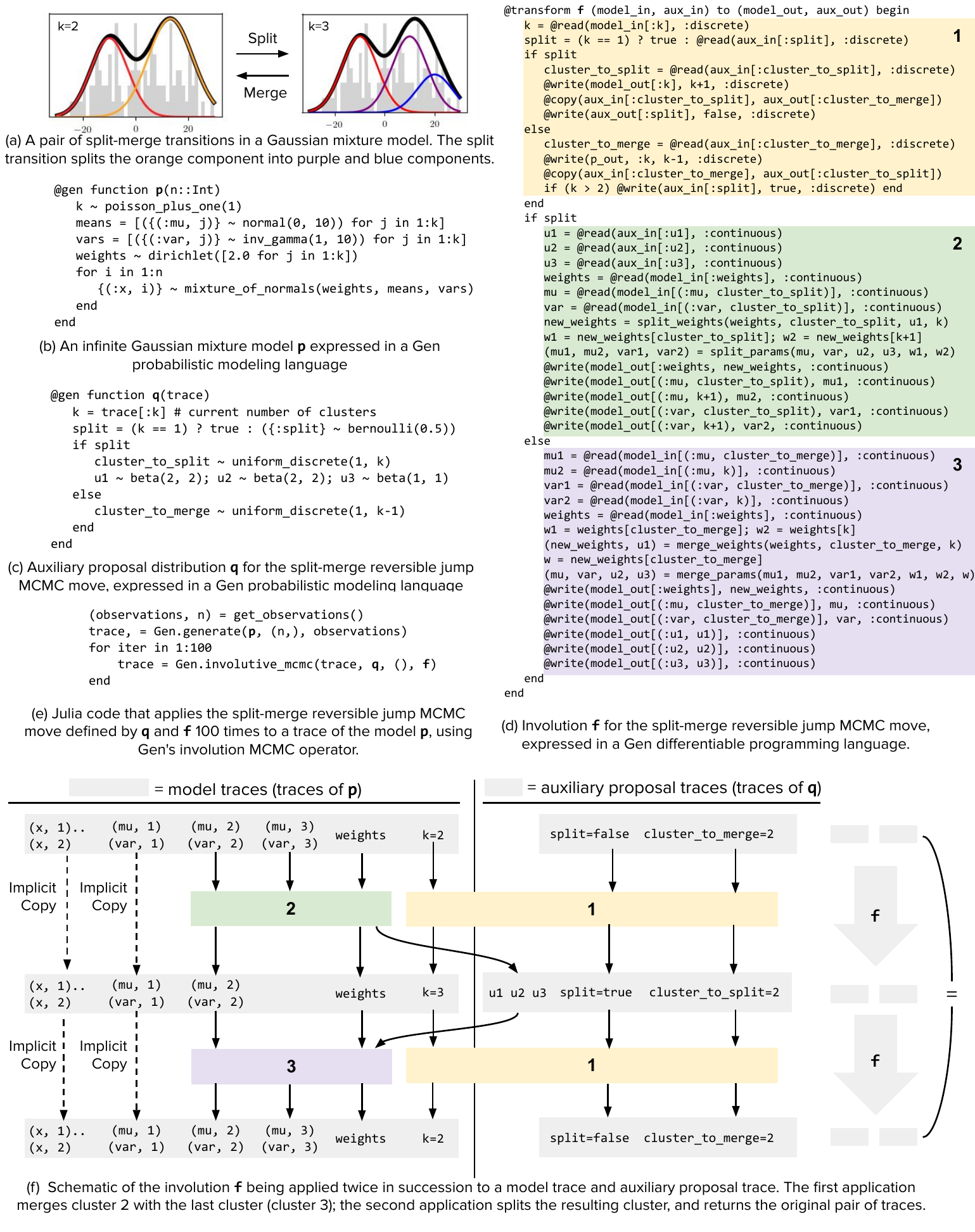}
    \caption{
Example of reversible jump MCMC~\citep{green1995reversible} implemented using involutive MCMC in Gen.
The example implements a `split-merge move' in a infinite Gaussian mixture model~\citep{richardson1997bayesian} using three Gen programs:
(1) a probabilistic program $\mathtt{p}$ encoding the generative model (shown in b),
(2) a probabilistic program $\mathtt{q}$ encoding an auxiliary probability distribution (shown in c),
and (3) a differentiable program $\mathtt{f}$ that encodes an involution on the space of pairs of traces of $\mathtt{p}$ and $\mathtt{q}$ (shown in d).
Gen's involutive MCMC operator (shown in e) automatically computes the acceptance probability.
}
    \label{fig:mixture}
\end{figure*}

\begin{figure*}[ht!]
    \centering
    \includegraphics[width=0.95\textwidth]{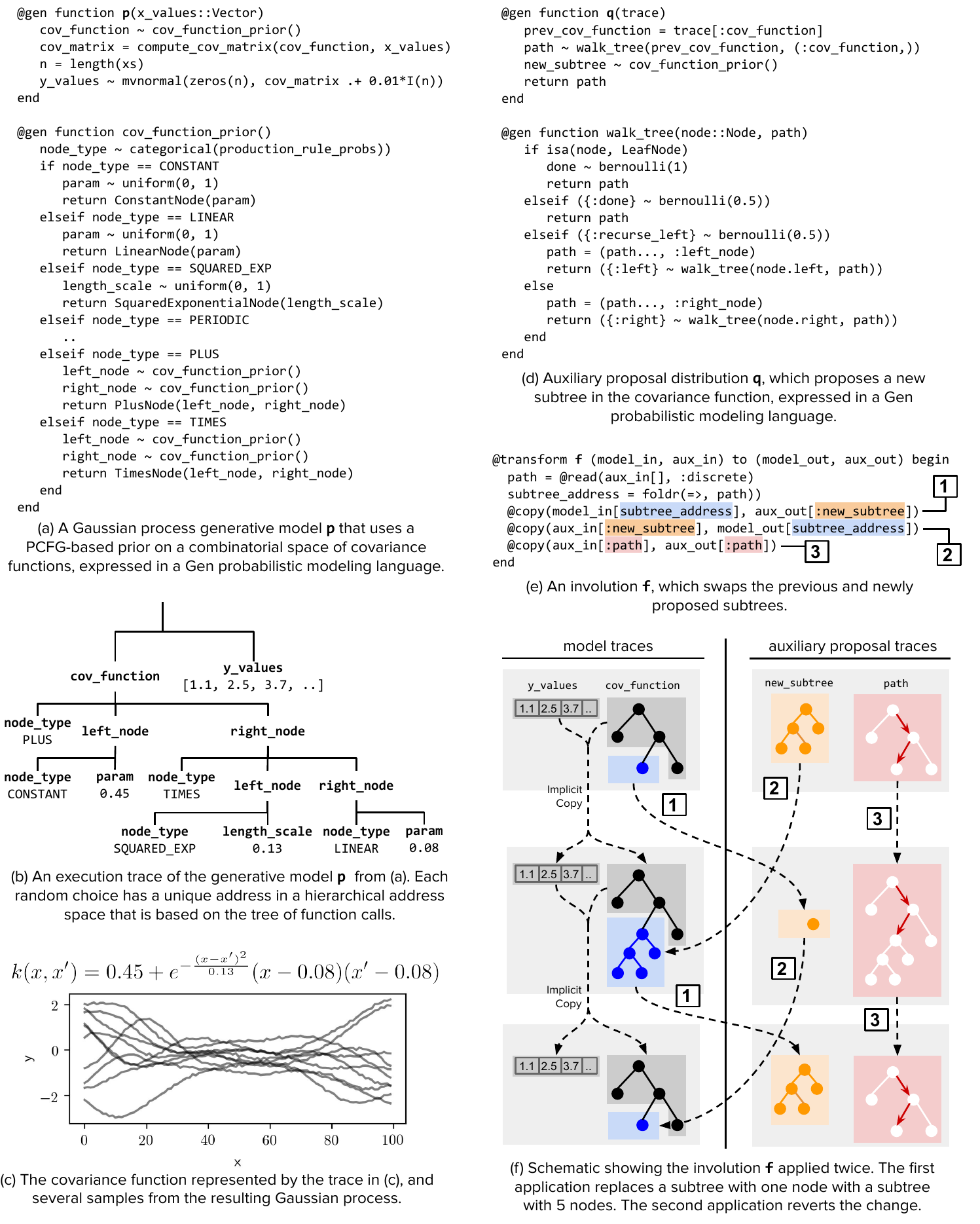}
    \caption{
A mixture kernel implemented using involutive MCMC in Gen, applied to infer the covariance function of a Gaussian process.
The prior on covariance functions is based on a probabilistic context-free grammar.
Each component kernel in the mixture replaces a subtree of the covariance function parse tree with a new subtree.
The mixture kernel chooses a random subtree to replace via a random walk on the parse tree.
The mixture kernel is composed from three Gen programs:
(1) a probabilistic program $\mathtt{p}$ encoding the generative model (shown in a),
(2) a probabilistic program $\mathtt{q}$ encoding an auxiliary probability distribution (shown in b), and
(3) a differentiable program $\mathtt{f}$ that encodes an involution (shown in d).
}
    \label{fig:structure-learning}
\end{figure*}

\clearpage
\section{INVOLUTIVE MCMC ON GENERAL STATE SPACES} 

Involutive MCMC is a general framework for constructing MCMC kernels that are stationary
for a target probability distribution $p$. Informally, the algorithm works as follows: starting at some state $x$, we first sample
an auxiliary variable $y \sim q_x$ from a state-dependent auxiliary distribution. We then apply an \textit{involution} $f$ to the pair $(x, y)$ to obtain $(x', y')$. 
Finally, we compute an acceptance probability $\alpha$ and either accept $x'$ as the new state, or reject it and repeat the previous state $x$.
Different choices of $q$ and $f$ recover many algorithms from the literature~\citep{imcmc}.

In this section, we present involutive MCMC for models $p$ and auxiliary kernels
$q$ defined over general state spaces.
By emphasizing general state spaces, we intend to clarify
a potential point of confusion regarding the involutive MCMC algorithm: as presented by
\citet{imcmc},
the acceptance probability $\alpha$ depends on the Jacobian of the involution $f$, but
it is not immediately clear how to define this Jacobian when $f$ may operate on samples
from arbitrary measurable spaces, rather than on vectors in $\mathbb{R}^n$. Our 
reformulation of the algorithm below is general enough to handle arbitrary model and auxiliary
distributions, and precise enough to enable automation via probabilistic and differentiable
programming: the rest of this paper uses it to develop a technique for deriving efficient implementations
of involutive MCMC algorithms automatically, 
given only declarative specifications of $p$, $q$, and $f$.

\subsection{General Involutive MCMC}

Let $(X, \Sigma_P, \mu_P)$ and $(Y, \Sigma_Q, \mu_Q)$ denote two general measure spaces with $\sigma$-finite $\mu_P$ and $\mu_Q$. 
Involutive MCMC (Algorithm~\ref{alg:involutive-mcmc}) implements a transition kernel that is invariant for a \textit{model distribution} given by $p : X \to [0, \infty)$, a probability density over $X$ with respect to $\mu_P$. Each iteration, the algorithm first samples auxiliary variables $y \in Y$ from an \emph{auxiliary distribution} $q_x$ based on the model's current state $x$: for each $x \in X$ such that $p(x) > 0$, $q_x : Y \to [0, \infty)$ is a probability density with respect to $\mu_Q$.

The resulting pair $(x, y)$ of the current model state and the newly sampled auxiliary state will be an element of the joint space $Z := \{(x, y) \in X \times Y \mid p(x) q_x(y) > 0 \}$. We can equip $Z$ with the $\sigma$-algebra 
$\Sigma := \{A \cap Z \mid A \in \Sigma_P \otimes \Sigma_Q\}$ (assuming $Z$ is a $\mu_P \times \mu_Q$-measurable set), and a reference measure $\mu(A) := (\mu_P \times \mu_Q)(A)$. 

Let $f : Z \to Z$ denote an involution ($f^{-1} = f$) such that the pushforward of $\mu$ under $f$, denoted $\mu \circ f^{-1}$, is absolutely continuous with respect to $\mu$, with Radon-Nikodym derivative $d (\mu \circ f^{-1}) / d\mu : Z \to [0, \infty)$. Involutive MCMC runs $f$ on $(x, y)$ to obtain $(x', y')$, then computes an acceptance probability $\alpha$.
With probability $\alpha$, the new state $x'$ is returned; otherwise, the previous state $x$ is repeated.

\begin{algorithm}[h]
\begin{algorithmic}
\Procedure{involutive-mcmc}{$p$, $q$, $f$, $x$}
    \State $y \sim q_x(\cdot)$ \Comment{Sample auxiliary state}
    \State $(x', y') \gets f(x, y)$ \Comment{Apply involution}
    \State $\alpha \gets
        \displaystyle \frac{p(x') q_{x'}(y')}{p(x) q_{x}(y)} \cdot \left( \frac{d (\mu \circ f^{-1})}{d \mu} (x, y)\right)$
    \State $r \sim \mathrm{Uniform}(0, 1)$
    \State \algorithmicif \, $r \le \alpha$ \algorithmicthen \, \Return $x'$ \algorithmicelse \, \Return $x$ 
\EndProcedure
\end{algorithmic}
\caption{Involutive MCMC}
\label{alg:involutive-mcmc}
\end{algorithm}

\begin{theorem}[Involutive MCMC is stationary]
Involutive MCMC defines a probability kernel $k$ on $X$ that is stationary with respect to the model probability distribution.
That is, $\int_X k_x(B) p(x) d\mu_X(dx) = \int_B p(x) d\mu_X(dx)$ for all $B \in \Sigma_X$.
\end{theorem}
\begin{proof}
The proof is presented in stages in the appendix (see Section~\ref{sec:involution-detailed-balance}, Section~\ref{sec:involution-is-stationary}, and Section~\ref{sec:involutive-mcmc-is-stationary}).
\end{proof}

\subsection{Probability Distributions on Dictionaries}
\label{sec:dists-on-dicts}
While maximally general, the measure-theoretic formulation of involutive MCMC in Algorithm~\ref{alg:involutive-mcmc} is not amenable to an automated implementation, because it does not indicate how to compute the Radon-Nikodym derivative that is required for the acceptance probability, and it is unclear how to specify the probability measures involved.

While restricting the state space to vectors of real numbers would address these issues, we seek a 
representation that remains flexible enough to represent complex hybrid state spaces with numeric, symbolic, and structure uncertainty.
Therefore, we use state spaces consisting of 
\textit{finite dictionaries} that map (possibly random) keys to (possibly random) values. 
Dictionaries include vectors as a special case (a vector $\mathbf{x} \in \mathbb{R}^n$ can
be represented as a dictionary mapping the keys $1, \dots, n$ to the values $x_1, \dots, x_n$),
but are more flexible: different keys can hold values of different types (e.g. integers, strings), and we
can also consider distributions in which the set of keys is itself random, which is useful for model
selection problems and structure uncertainty more generally.

This section describes probability distributions on dictionaries, and gives a constructive definition of the involutive MCMC acceptance probability in this setting in terms of a Jacobian.
Section~\ref{sec:automating} will then show how probability distributions on dictionaries can be specified with probabilistic programs, and how probabilistic programming techniques can automatically compute probability densities on spaces of dictionaries.

\textbf{The space of finite dictionaries.}
We fix a countably infinite set $\mathcal{K}$ of possible keys, such that each key $k$ is either called \emph{discrete} ($k \in \mathcal{I}$) or \emph{continuous} ($k \in \mathcal{J}$), where $\mathcal{K} = \mathcal{I} \cup \mathcal{J}$.\footnote{
It is possible to assign a general measure space to each key, but this is not necessary for our purposes.}
Let $V_k$ denote the set of possible values for key $k$, where $V_k$ is a countable set for each discrete key, and where $V_k = \mathbb{R}^{d_k}$ for each continuous key for some $d_k$.
Given a set of keys $K$, let $V_K = \times_{k \in K} V_k$ denote the set of \textit{assignments} of values to each key.
Then the set of all finite dictionaries is $\mathcal{D} := \bigcup_{K \subset \mathcal{K}, |K| < \infty} \{(K, \mathbf{x}) \mid \mathbf{x} \in V_K\}$. That is, a dictionary specifies 
a finite set of keys $K \subset \mathcal{K}$ at which it has values, and an assignment $\mathbf{x}$ of values $x_k$ for each.

\textbf{Relationship to representation of~\citet{green1995reversible}.}
\citet{green1995reversible} uses a state space that is the countable union of `models', where each model is typically a vector of real-valued parameters.
Dictionaries have substantially more structure:
Instead of monolithic `models', dictionaries use a more elaborate discrete state that includes the set of keys and the assignment to the discrete keys. 
Also, because continuous keys play the role of real-valued parameters, it is possible to express that a given real-valued parameter is shared between models.
The additional structure of dictionaries enables the automation techniques in Section~\ref{sec:automating}.

\textbf{A measure space of finite dictionaries.}
We associate a measure $\mu_k$ on $V_k$ with each key $k$---the counting measure for each discrete key and the Lebesgue-measure on $\mathbb{R}^{d_k}$ for each continuous key.
For each finite set of keys $K$, we make $V_K$ a measure space using the standard product $\sigma$-algebra $\Sigma_K = \otimes_{k \in K} \Sigma_k$ and the product measure $\mu_K = \times_{k \in K} \mu_k$.
We equip $\mathcal{D}$ with the $\sigma$-algebra $\Sigma_{\mathcal{D}} := \{\bigcup_{K \subset \mathcal{K}, |K| < \infty} \{(K, \mathbf{x}) \mid \mathbf{x} \in B_K\} \mid B_K \in \Sigma_K \, \text{for each finite } K \subset \mathcal{K}\}$ to obtain a measurable space
of dictionaries.
A reference measure $\mu_\mathcal{D}$ on this space can be constructed using the product measures $\mu_K$: we set $\mu(B) := \sum_{K \in \mathcal{K}, |K| < \infty} \mu_K(\{\mathbf{x} \mid (K, \mathbf{x}) \in B\})$.

\textbf{Notation for dictionaries.} Given a dictionary $m = (K, \mathbf{x})$, we write $K_m$ for $K$ and $m[k]$ for the value $x_k$ associated with a key $k \in K$.
We also denote specific dictionaries using notation $\{k_1 \mapsto v_1, k_2 \mapsto v_2, \ldots\}$.
For example, the dictionary $(K, \mathbf{x})$ with $K = \{1, \mathtt{"foo"}\}$ and $x_1 = 0.123$ and $x_{\mathtt{"foo"}} = 5$ is denoted 
$\{1 \mapsto 0.123, \mathtt{"foo"} \mapsto 5\}$.

\textbf{Probability distributions on finite dictionaries.} 
When $V_k$ is discrete for all keys $k$, a probability distribution on dictionaries is defined by a probability mass function $p : \mathcal{D} \to [0, 1]$ that assigns a probability $p(m)$ to each dictionary $m \in \mathcal{P}$ such that $\sum_{m \in \mathcal{D}} p(m) = 1$.
More generally a probability distribution on dictionaries is defined by a probability density $p : \mathcal{D} \rightarrow [0, \infty)$ such that $\int p(m) \mu_\mathcal{D}(\text{d}m) = 1$. 
The probability mass is distributed among the finite sets of keys $K \subseteq \mathcal{K}$:
\begin{equation} \label{eq:integral}
1 = \sum_{\substack{K \subseteq \mathcal{K}\\|K| < \infty}}
\left( \sum_{\mathbf{x}_1} \left( \int_{\mathbb{R}^{d_K}} p((K, (\mathbf{x}_1, \mathbf{x}_2))) d\mathbf{x}_2 \right) \right)
\end{equation}
where $d_K : = \sum_{k \in K \cap \mathcal{J}} d_{k}$ is the total continuous dimension for keys $K$, and where $\mathbf{x}_1$ is an assignment to the discrete choices in $K$ and $\mathbf{x}_2$ is an assignment to the continuous choices.

We now give an example to build intuition.
Consider a generative model of univariate data points $y_1, \dots, y_n$ from a Gaussian mixture with an unknown number of components $k$, each with unknown mean $m_i$ and variance $s_i$. If we place a Gaussian prior on $m_i$, an inverse Gamma prior on $s_i$, and a Poisson prior on $k$, the resulting density on dictionaries $d$ is:
\begin{align*}
p(d) =
\begin{array}{l}
p_{\mathrm{poisson}(3)}(d[\mathtt{k}]) \cdot\\
\prod_{i=1}^{d[\mathtt{k}]} p_{\mathrm{normal}(0, 1)}(d[\mathtt{m}_i]) \cdot\\
\prod_{i=1}^{d[\mathtt{k}]} p_{\mathrm{inversegamma}(1, 10)}(d[\mathtt{s}_i]) \cdot\\
\prod_{i=1}^n  \frac{1}{d[\mathtt{k}]} \sum_{j=1}^{d[\mathtt{k}]} \cdot p_{\mathrm{normal}(d[\mathtt{m}_j], d[\mathtt{s}_i])}(d[\mathtt{y}_i])
\end{array}
\end{align*}
when $K_d = \{\texttt{k}, \texttt{y}_\texttt{1}, \dots, \texttt{y}_\texttt{n}\} \cup \{\texttt{m}_i \mid 1 \leq i \leq d[\texttt{k}]\} \cup \{\texttt{s}_i \mid 1 \leq i \leq d[\texttt{k}]\}$, and 0 otherwise. In this case, we have $V_\texttt{k} = \mathbb{N}$ with the counting measure for $\mu_\texttt{k}$, and for all other keys $k \in \mathcal{K}$, $V_k = \mathbb{R}$ with the Lebesgue measure for $\mu_k$.
For each $j \in \{0, 1, \ldots\}$, the probability mass assigned to key set
$\{\texttt{k}, \texttt{y}_\texttt{1}, \dots, \texttt{y}_\texttt{n}\} \cup \{\texttt{m}_i \mid 1 \leq i \leq j\} \cup \{\texttt{s}_i \mid 1 \leq i \leq j\}$
is $p_{\mathrm{poisson}(s)}(j)$ (via Equation~(\ref{eq:integral})).

\textbf{Conditional distributions via disintegration.} Consider a probability density $p$ on the space $\mathcal{D}$ of dictionaries. 
We say a key $k \in \mathcal{K}$ \textit{almost always appears} if $\int \mathbf{1}[k \in K_m] p(m) \mu_\mathcal{D}(\text{d}m) = 1$. Suppose $B$ is a set of keys that almost always appear for $p$, and that $b = (B, \mathbf{b})$ is a dictionary with keys $B$. Furthermore, let $(K_m, \mathbf{m}) \oplus (K_n, \mathbf{n}) := (K_m \cup K_n, (\mathbf{m}, \mathbf{n}))$ be the \textit{merge} of two dictionaries $m$ and $n$ defined on disjoint key sets $K_m$ and $K_n$. Then we can define the conditional density $p(d \mid b) := \mathbf{1}[K_d \cap B = \emptyset] \frac{p(d \oplus b)}{\int_{\{m \mid K_m \cap B = \emptyset\}} p(m \oplus b) \mu_\mathcal{D}(\text{d}m)}$ when the denominator is finite.
If each each key $k \in B$ is discrete (e.g. $\mu_k$ is the counting measure), then this definition corresponds to the ordinary notion of conditioning on an event (namely, the event that a sample from $p$ agrees with the dictionary $b$ on all keys in $B$). When this is not the case, it corresponds to a more general measure-theoretic notion called disintegration~\citep{chang1997conditioning}. 

Consider the infinite univariate mixture model, and the conditional density given observed data $(B, \mathbf{b}) := \{\mathtt{y}_1 \mapsto y_1, \ldots, \mathtt{y}_{\mathtt{n}} \mapsto y_{\mathtt{n}}\}$.
The conditional density $p(d | b)$ is nonzero only if $K_d = \{\texttt{k}\} \cup \{\texttt{m}_i \mid 1 \leq i \leq k\} \cup \{\texttt{s}_i \mid 1 \leq i \leq k\}$ for some $k$ ($d$ does not contain y-values), and the denominator in the definition of $p(d | b)$ simplifies to the familiar sum of marginal likelihoods over all $k$, where each marginal likelihood is a Riemann integral over $\mathbb{R}^{2k}$.

\subsection{Involutive MCMC with Dictionaries}


Suppose that the model distribution and auxiliary distributions are probability distributions on dictionaries, with densities $p$ and $q_x$.
Then, $X$ and $Y$ are both sets of dictionaries, and the joint space $Z$ is a set of pairs $(x, y)$ of dictionaries with keys $\mathcal{K}_P$ and $\mathcal{K}_Q$ respectively, so that $Z = X \times Y \subseteq \mathcal{D}_P \times \mathcal{D}_Q$ where $\mathcal{D}_P$ is the set of dictionaries on keys taken from $\mathcal{K}_P$ and similarly for $\mathcal{D}_Q$.
To simplify the notation, and without loss of generality, we will assume that $\mathcal{K}_P$ and $\mathcal{K}_Q$ are disjoint\footnote{If $\mathcal{K}_P$ and $\mathcal{K}_Q$ are not disjoint then, they can be made so by adding a different prefix to the keys of each set.}, and we define $Z := \{ x \oplus y : p(x)q_x(y) > 0\} \subseteq \mathcal{D}$,
where $\mathcal{D}$ is the set of dictionaries on keys from $\mathcal{K}_P \cup \mathcal{K}_Q$
(recall $x \oplus y$ denotes the dictionary resulting from merging dictionaries $x$ and $y$ with disjoint keys).

Suppose there is a countable partition of $Z$ into $\{Z_e : e \in E\}$ such that if $(K_1, \mathbf{x}_1)$ and $(K_2, \mathbf{x}_2)$ are two dictionaries in the same component $Z_e$, then $K_1 = K_2$ and they agree on all \textit{discrete} values: $x_{1k} = x_{2k}$ for all $k \in K_1 \cap \mathcal{I}$.
Then each set $Z_e$ is isomorphic to a Euclidean space of assignments to the continuous keys in the two dictionaries.
Suppose there is an involution $g : E \to E$ between elements of the partition, and a family of continuously differentiable bijections $h_{e} : Z_e \to Z_{g(e)}$ indexed by $e \in E$, with $h_{e} = h_{g(e)}^{-1}$.
Let $e(z) \in E$ denote which element of the partition a dictionary $z \in Z$ belongs to.
Then $f : Z \to Z$ given by $f(z) := h_{e(z)}(z)$ is an involution:
\begin{equation*}
f(f(z)) = h_{e(f(z))}(h_{e(z)}(z)) = h_{g(e(z))}(h_{e(z)}(z)) = z.
\end{equation*}
Let $|J h_e|(z)$ denote the absolute value of the determinant of the Jacobian of $h_e$, evaluated at $z$.
Then, the acceptance probability in Algorithm~\ref{alg:involutive-mcmc} simplifies to:
\begin{equation} \label{eq:acceptance-ratio-dictionaries}
\frac{p(x') q_{x'}(y')}{p(x) q_{x}(y)} \cdot |J h_e|(z)
\end{equation}
One example of a valid partition of $Z$ is given by equivalence classes of the following equivalence relation:
\[
z_1 \sim z_2 \iff (K_{z_1} = K_{z_2}) \land (z_1[k] = z_2[k] \; \forall k \in K_{z_1} \cap \mathcal{I})
\]
(dictionaries are equivalent if they contain the same keys and they agree on the value of all discrete keys).
See Section~\ref{sec:radon-nikodym-special-case} of the appendix for details.

\begin{algorithm*}[h!]
\begin{algorithmic}
\Require probabilistic programs $\mathcal{P}$ and $\mathcal{Q}$; differentiable program $\mathcal{F}$; initial state $x$; observations $b$
\vspace{1mm}
\Procedure{auto-involutive-mcmc}{$\mathcal{P}$, $\mathcal{Q}$, $\mathcal{F}$, $x$, $b$}
    \State\vspace{-2mm}
    \LineComment{Sample $y \sim q_x(\cdot)$ and compute log-density via probabilistic program $\mathcal{Q}$}
    \State $(y, \log q_x(y)) \gets \textproc{trace-and-score}(\mathcal{Q}_x)$
    \State\vspace{-2mm}
    \LineComment{{\small Compute $(x' \oplus y') = f(x \oplus y)$ and log(Radon-Nikodym derivative) via differentiable program $\mathcal{F}_b$}}
    \State $(x' \oplus y', \log D) \gets \textproc{run-involution}(\mathcal{F}, x \oplus y, b)$
    \State\vspace{-2mm}
    \LineComment{Compute log-density via probabilistic program $\mathcal{P}$}
    \State $\log \tilde{p}(x \oplus b) \gets \textproc{score}(\mathcal{P}, x \oplus b)$
    \State\vspace{-2mm}
    \LineComment{Compute log-density via probabilistic program $\mathcal{P}$}
    \State $\log \tilde{p}(x' \oplus b) \gets \textproc{score}(\mathcal{P}, x' \oplus b)$ 
    \State\vspace{-2mm}
    \LineComment{Compute log-density via probabilistic program $\mathcal{Q}$}
    \State $\log q_{x'}(y') \gets \textproc{score}(\mathcal{Q}_{x'}, y')$ 
    \State\vspace{-2mm}
    \LineComment{Compute acceptance probability and accept or reject}
    \State $\alpha \gets \min\{1, \exp(\log \tilde{p}(x' \oplus b) - \log \tilde{p}(x \oplus b) + \log q_{x'}(y') - \log q_x(y) + \log D)\}$ 
    \State \textbf{with probability} \,$\alpha$ \,\Return $x'$ \algorithmicelse \, \Return $x$  \algorithmicend
\EndProcedure\\
\vspace{1mm}
\begin{minipage}[t]{0.49\linewidth}
\Procedure{trace-and-score}{$\mathcal{P}$}
   \State $s \gets 0$
   \State $x \gets \{\}$
   \execute[$\mathcal{P}$], but with $\mathtt{"}k \sim \mathrm{distribution}\mathtt{"} \equiv$ (
        \State 1. set $v \sim \mathrm{distribution}$
        \State 2. set $x[k] \gets v$
        \State 3. set $s \gets s + \textproc{logpdf}(\mathrm{distribution}, v)$
        \State 4. evaluate to $v$)
    \END
   \State \Return $(x, s)$
\EndProcedure
\end{minipage}%
\hfill
\begin{minipage}[t]{0.49\linewidth}
\Procedure{score}{$\mathcal{P}$, $x$}
    \State $s \gets 0$
    \State $K \gets\{\}$
    \execute[$\mathcal{P}$], but with $\mathtt{"}k \sim \mathrm{distribution}\mathtt{"} \equiv$ (
        \State 1. set $v \gets x[k]$
        \State 2. set $K \gets K \cup \{k\}$
        \State 3. set $s \gets s + \textproc{logpdf}(\mathrm{distribution}, v)$
        \State 4. evaluate to $v$)
    \END
   \State \algorithmicif\hspace{1mm} $K \ne K_x$ \algorithmicthen\hspace{1mm} $s \gets -\infty$ \algorithmicend
   \State \Return $s$
\EndProcedure
\end{minipage}\\
\vspace{1mm}
\Procedure{run-involution}{$\mathcal{F}$, $z$, $b$}
    \LineComment{Run the involution, and keep track of copied, read, and written continuous keys}
    \State $z' \gets \{\}$ \Comment{Initialize empty output dictionary}
    \execute[$\mathcal{F}$], but with
        \State $\mathtt{"@write}(k, v)\mathtt{"} \equiv$
            (\algorithmicif\hspace{1mm} $k \in \mathcal{I}$
            \algorithmicthen\hspace{1mm} set $z'[k] \gets v$
            \algorithmicelse\hspace{1mm} set $z'[k] \gets v$ and set $W \gets W \cup \{k\})$
        \State $\mathtt{"@read}(k)\mathtt{"} \equiv$
            (\algorithmicif\hspace{1mm} $k \in \mathcal{I}$
            \algorithmicthen\hspace{1mm} evaluate to $z[k]$
            \algorithmicelse\hspace{1mm} set $R \gets R \cup \{k\}$ and evaluate to $z[k]$)
        \State $\mathtt{"@copy}(k_1, k_2)\mathtt{"} \equiv$ (set $z'[k_2] \gets z[k_1]$ and set $C \gets C \cup \{k_1\}$)
        \State (and reading from observations $b$ for any $k \in K_b$)
    \END
    \LineComment{Use automatic differentiation of $\mathcal{F}$ to compute Jacobian, skipping copied addresses}
    \For{$i$ in $1$ to $|W|$}
        \State $k_i \gets W_i$ \Comment{Pick $i$th key in $W$; the order does not matter}
        \State $J[:,i] \gets \nabla_{\mathbf{z}_{R \setminus C}} \left( h_{e(z)}[a_i] \right)$ \Comment{{\small Gradient of $z'[k_i] = h_{e(z)}[k_i]$ w.r.t. non-copied continuous inputs ($z_k$ for $k \in R \setminus C$)}}
    \EndFor
    \State \Return $(z', \log(|\det(J)|))$
\EndProcedure
\end{algorithmic}
\caption{Automated Involutive MCMC}
\label{alg:auto-involutive-mcmc}
\end{algorithm*}

\section{AUTOMATING INVOLUTIVE MCMC WITH TRACES} \label{sec:automating}

Involutive MCMC is a general framework that can be used to develop diverse MCMC algorithms
for models over arbitrary state spaces. We wish to \textit{automate} the implementation
details for involutive MCMC algorithms, given only a specification of the model $p$, the
auxiliary distribution $q$, and the involution $f$.
To do so, we require a representation for the distributions $p$ and $q$ that is flexible
enough to represent the full variety of models and auxiliary distributions of interest to
practitioners. The representation must support density evaluation and sampling. It is also desirable that the representation be \textit{structured}: the more
information available to us (e.g., about the decomposition of a distribution's state space into 
individual univariate and multivariate random variables, or about conditional independence relationships in a model),
the easier it will be for the implementation to exploit this structure automatically by using
more efficient data structures and low-level manipulations.

\subsection{Trace-Based Probabilistic Programming} \label{sec:probabilistic-programming}

\textit{Probabilistic programs} are flexible and structured representations for  probability distributions. Unlike densities (but like Bayesian networks), probabilistic programs can be efficiently sampled and contain explicitly represented information
about some conditional independence relationships in a model. 
But unlike Bayesian networks, they do not assume a fixed number of random variables, state space dimension, or dependency structure.

At the most basic level, a probabilistic program is a program that makes random choices. 
Any such program induces a probability distribution over its possible \textit{execution traces}, records of each random choice it makes.
If each random choice is associated with a unique \textit{address} from the set of
dictionary keys $\mathcal{K}$, 
then these traces can be viewed as finite dictionaries, mapping the address of each random choice to its value. 
The distribution induced by a probabilistic program over
its execution traces can thus be understood as a measure on the space $(\mathcal{D}, \Sigma_\mathcal{D}, \mu_\mathcal{D})$ introduced in the previous section. Furthermore, densities of trace distributions with respect to $\mu_\mathcal{D}$ are typically easy to compute.

In this section, we introduce a probabilistic programming language from the Gen probabilistic programming system~\citep{cusumano2019gen}, and present a technique for automating the implementation of involutive MCMC algorithms when the model $p$ and auxiliary distribution $q_x$ are both represented as
probabilistic programs in this language.
But the technique is not limited to the Gen system---we also provide the implementation of a minimal probabilistic programming language in PyTorch that supports the automation technique.

\subsubsection{A Probabilistic Programming Language}

Our probabilistic programming language augments the syntax of Julia~\citep{bezanson2017julia}
with a single new construct, the `$\{\mathtt{address}\} \sim \mathtt{distribution}$' expression, for making a \textit{named} random choice. An execution trace of a Gen probabilistic program is a dictionary that can be sampled by running the program according to Julia's usual semantics, and upon encountering an expression of the form `$\{\mathtt{address}\} \sim \mathtt{distribution}$', (i) evaluating the \textit{address expression} (`$\mathtt{address}$') to obtain an address $k \in \mathcal{K}$; (ii) evaluating the \textit{distribution expression} (`$\mathtt{distribution}$') to obtain a probability distribution over the measurable space $V_k$; (iii) sampling a value $x_k$ from this distribution and adding the mapping $\{k \mapsto x_k\}$ to the execution trace; and (iv) returning the sampled value $x_k$ to the program, to continue execution. When execution terminates, the execution trace has accumulated a mapping for each random choice encountered during the program's execution.

For example, consider the probabilistic program below, which defines a Gaussian mixture model with an unknown number of components:
\noindent
\begin{center}
\begin{tabular}{c}
{\begin{lstlisting}[basicstyle=\small\ttfamily]
@gen function p(n::Int)
  k ~ poisson_plus_one(1)
  means = [
    ({(:mu, j)} ~ normal(0, 10)) for j in 1:k]
  vars = [
    ({(:var, j)} ~ inv_gamma(1, 10)) for j in 1:k]
  weights ~ dirichlet([2.0 for j in 1:k])
  for i in 1:n
    {(:x, i)} ~ mixture_of_normals(
        weights, means, vars)
  end
end
\end{lstlisting}}
\end{tabular}
\end{center}
The program \texttt{p} accepts as input an integer \texttt{n}, a number of data points. Each $\texttt{n}$ defines
a distinct distribution over dictionaries.
The first line of the program samples a number of mixture components from a Poisson prior using the address \texttt{:k}. (This line could also be written \texttt{k = \{:k\} $\sim$ poisson\_plus\_one(1)}: the \textit{address} is the symbol \texttt{:k}, and the result of the choice is assigned to a Julia variable called \texttt{k}. Because this is a common pattern, Gen provides the syntactic sugar \texttt{x $\sim$ d} as shorthand for \texttt{x = \{:x\} $\sim$ d}.) The program then samples \texttt{k} means and \texttt{k} variances from Gaussian and inverse Gamma priors, respectively. Each of these 2\texttt{k} random choices has its own address, determined by the address expression preceding it. For example, the mean for the fourth mixture component (if \texttt{k} $\geq 4$) has address \texttt{(:mu, 4)}. The mixture weights are then sampled at the address \texttt{:weights} from a Dirichlet distribution, and $\texttt{n}$ data points are sampled at addresses \texttt{(:x, 1)}, \dots, \texttt{(:x, n)}. The random choice at address $\mathtt{:}\mathtt{k}$ is discrete with $V_{\texttt{:k}} = \integers_{\ge 1}$, and the other random choices are continuous. 
Overall, the program defines a distribution over traces with the following density over traces $(K, \mathbf{x})$ with respect to $\mu_\mathcal{D}$:
\begin{align*}
p(n)((K, \mathbf{x})) =
\begin{array}{l}
p_{\mathrm{poisson}(1)}(x_{\mathtt{:}\mathtt{k}}-1) \cdot\\
\prod_{i=1}^{x_{\mathtt{:}\mathtt{k}}} p_{\mathrm{normal}(0, 10)}(x_{(\mathtt{:}\mathtt{mu}, i)}) \cdot\\
\prod_{i=1}^{x_{\mathtt{:}\mathtt{k}}} p_{\mathrm{inversegamma}(1, 10)}(x_{(\mathtt{:}\mathtt{var}, i)}) \cdot\\
p_{\mathrm{dirichlet}(2.0,\dots,2.0)}(x_{\mathtt{:}\mathtt{weights}})\cdot\\
\prod_{i=1}^n \sum_{j=1}^{x_{\mathtt{:}\mathtt{k}}} x_{\mathtt{:}\mathtt{weights}}[j] \cdot\\
\;\;\;\;\;\;\;\;\;\;\;\;\;\;\;\;\;\; p_{\mathrm{normal}(x_{\mathtt{(:mu, j)}}, x_{(\mathtt{:}\mathtt{var}, j)})}(x_{(\mathtt{:}\mathtt{x}, i)})
\end{array}
\end{align*}
when $K$ contains exactly the addresses $\mathtt{:}\mathtt{k}, \mathtt{:}\mathtt{weights}, (\mathtt{:}\mathtt{x}, i)$ for $i = 1, \dots, n$, and $(\mathtt{:}\mathtt{mu}, j)$ and $(\mathtt{:}\mathtt{var}, j)$ for $j = 1, \dots, x_{\mathtt{:}\mathtt{k}}$; otherwise, the density is 0.

\paragraph{Automatic Sampling and Density Computation for Probabilistic Programs.}
Sampling traces from, and computing densities of dictionaries under, the distribution on traces
induced by a probabilistic program is straight-forward to do, using a standard technique
in probabilistic programming.

We illustrate this technique in the \textsc{trace-and-score} and \textsc{score} subroutines in Algorithm~\ref{alg:auto-involutive-mcmc}.
The \textsc{trace-and-score} subroutine samples a trace by running a probabilistic program,
but recording the value of every encountered choice into a dictionary, which it returns once execution has terminated. It also returns the log density of the trace, calculated by accumulating densities of the individual choices it encounters. (Note that this procedure is only valid if the program
halts with probability 1. Otherwise, it could loop infinitely, and even if it terminates,
the density will be incorrect.) 

The density of an arbitrary dictionary $(K, \mathbf{x})$ under any probabilistic program's distribution on traces can be computed using \textsc{score}. The idea is to run the probabilistic program, and whenever a random choice `\texttt{\{address\} $\sim$ distribution}' is encountered, to look up the value $x_k$ in the dictionary, compute its density under the primitive distribution $d$, multiply this density into a running total, and return control to the probabilistic program as if the sampling instruction had executed and returned $x_k$. At the end, the running total can be returned as the trace's density. If at any point an address $k \not\in K$ is encountered, or if not all addresses in $K$ have been visited at the end of execution, the algorithm returns $0$ as the density.

\subsubsection{Automatically Computing Density Ratios in Involutive MCMC}

We can use two probabilistic programs $\mathcal{P}$ and $\mathcal{Q}$ to specify the model density $p$ and auxiliary densities $q_x$ that appear in involutive MCMC
(Algorithm~\ref{alg:involutive-mcmc}).\footnote{The programs must satisfy a mild technical requirement for our formalism to go through; see Section~\ref{sec:technical-requirement}}
In this case, the program $\mathcal{Q}$ 
accepts a trace $x$ of $\mathcal{P}$ as input, we denote a probabilistic program $\mathcal{Q}$
applied to input $x$ by $\mathcal{Q}_x$.

Typically, we cannot use a probabilistic program $\mathcal{P}$ to represent the target distribution directly: probabilistic programs implement simulators for a distribution, but target densities are typically not tractable to simulate (hence the need for MCMC). Instead, we may wish to sample from a target $p$ that arises from \emph{conditioning} a probabilistic program $\mathcal{P}$'s distribution over traces on observations of the values at some addresses. Let $\tilde{p}$ denote $\mathcal{P}$'s distribution over traces, and let $b = (B, \textbf{b})$ be a dictionary of observations, as described in Section~\ref{sec:dists-on-dicts}. Then, as shown in that section, the target distribution $p(x) = \tilde{p}(x \mid b) = \frac{\tilde{p}(x \oplus b)}{\mathcal{L}(b)}$, where $\mathcal{L}(b) = \int_{\{m \mid m \cap B = \emptyset\}} p(m \oplus b) \mu_\mathcal{D}(\text{d}m)$ is the marginal likelihood of $b$, and does not depend on $x$. Then $\tilde{p}(x \oplus b) = \mathcal{L}(b)p(x)$. If we use $\tilde{p}(x \oplus b)$ in place of $p(x)$ to compute the ratio of densities in Algorithm~\ref{alg:involutive-mcmc}, we will wind up with the same output, because $\mathcal{L}(b)$ will cancel in the numerator and denominator:
\begin{equation}
\frac{p(x') q_{x'}(y')}{p(x) q_{x}(y)}
= \frac{\tilde{p}(x' \oplus b) q_{x'}(y')}{\tilde{p}(x \oplus b) q_x(y)}
= \frac{\tilde{p}(x' | b) q_{x'}(y')}{\tilde{p}(x | b) q_x(y)}
\end{equation}

Thus, this ratio can be computed term-by-term using \textsc{score} on the probabilistic programs $\mathcal{P}$ and $\mathcal{Q}$ (Algorithm~\ref{alg:auto-involutive-mcmc}): to compute $q_x(y')$, we run the algorithm directly on the program $\mathcal{Q}_x$, and to compute $p(x)$ and $p(x')$, we actually merge the dictionary $x$ with the observations $b$ and compute $\mathcal{L}(b)p(x) = \tilde{p}(x \oplus b)$ instead, by running the algorithm on $\mathcal{P}$ with trace $x \oplus b$. The density $q_x(y)$ can be computed while $y$ is being sampled, using \textsc{trace-and-score}.

Section~\ref{sec:sparsity} describes a more efficient approach that exploits sparsity in the involution and cancellations in the acceptance ratio for improved efficiency, but requires a more sophisticated probabilistic programming runtime system.

Figure~\ref{fig:mixture}b and Figure~\ref{fig:mixture}c show examples of probabilistic programs $\mathcal{P}$ and $\mathcal{Q}$ respectively, for a split-merge reversible jump move.

\subsection{Differentiable Programming with Traces} \label{sec:differentiable-programming-traces}

Section~\ref{sec:probabilistic-programming} showed that if the densities $p$ and $q$ are specified using probabilistic programs $\mathcal{P}$ and $\mathcal{Q}$, then the density ratio in the acceptance probability for involutive MCMC on dictionaries (Equation~(\ref{eq:acceptance-ratio-dictionaries})) can be automated using probabilistic programming techniques.
This section shows that if the involution $f$ is specified using a differentiable program $\mathcal{F}$ that transforms the execution traces of probabilistic programs, then the Jacobian factor in Equation~(\ref{eq:acceptance-ratio-dictionaries}) can also be automated, using automatic differentiation.
The procedure \textproc{auto-involutive-mcmc} in Algorithm~\ref{alg:auto-involutive-mcmc} combines these two ideas and automates involutive MCMC given the programs $\mathcal{P}$, $\mathcal{Q}$, and $\mathcal{F}$.


\subsubsection{A Differentiable Programming Language for Manipulating Traces}
Recall that for involutive MCMC on a state space of dictionaries, we define $Z := \{x \oplus y : p(x) q_x(y) > 0\} \subseteq \mathcal{D}$, where $x \oplus y$ is the dictionary resulting from merging dictionaries $x$ and $y$ with disjoint keys.
The involution is a function $f : Z \to Z$.

We now introduce a simple differentiable programming language for specifying involutions $f$.
The language needs to have syntax for reading the value from an address in $x \oplus y \in Z$ and writing to an address in $x' \oplus y' \in Z$.
To read a value $(x \oplus y)[a]$ at address $a$, we use the \texttt{@read} keyword:
\noindent
\begin{center}
\begin{tabular}{c}
{\begin{lstlisting}[basicstyle=\small\ttfamily]
    value = @read(<address>, <type>)
\end{lstlisting}}
\end{tabular}
\end{center}
The first argument is the address $a$ and the second argument is either $\mathtt{:}\mathtt{discrete}$ or $\mathtt{:}\mathtt{continuous}$, and informs the interpreter whether the random choice at that address is drawn from a discrete or continuous distribution (this information will be used to support automatic differentiation).
Recall that $x$ is the trace of the model probabilistic program, and $y$ is the trace of the auxiliary probabilistic program.
Each address $a$ therefore needs to specify which of these traces to read from, and the address within that trace.
The traces are given names in the function signature:
\noindent
\begin{center}
\begin{tabular}{c}
\begin{lrbox}{\mybox}%
\begin{lstlisting}[basicstyle=\small\ttfamily]
@transform f (model_in, aux_in) to (model_out, aux_out)
begin
..
end
\end{lstlisting}
\end{lrbox}%
\scalebox{0.94}{\usebox{\mybox}}
\end{tabular}
\end{center}
Here, the traces $x$, $y$, $x'$ and $y'$ are given names $\mathtt{model\_in}$, $\mathtt{aux\_in}$, $\mathtt{model\_out}$, and $\mathtt{aux\_out}$, respectively.
The syntax for address $\mathtt{a}$ within trace $\mathtt{trace}$ is $\mathtt{trace[a]}$.
For example, to read the value of a continuous address $\mathtt{:}\mathtt{a}$ from the input model trace ($x$):
\noindent
\begin{center}
\begin{tabular}{c}
{\begin{lstlisting}[basicstyle=\small\ttfamily]
val = @read(model_in[:a], :continuous)
\end{lstlisting}}
\end{tabular}
\end{center}
The syntax for writing to an address in $x' \oplus y' \in Z$ is similar.
For example to write a value $\mathtt{val}$ to $x'[\mathtt{:}\mathtt{a}]$:
\noindent
\begin{center}
\begin{tabular}{c}
{\begin{lstlisting}[basicstyle=\small\ttfamily]
@write(model_out[:a], val, :continuous)
\end{lstlisting}}
\end{tabular}
\end{center}
Note that the input traces $x, y$ are distinct from the output traces $x', y'$;
input traces can only be read from, and output traces can only be written to.
For example, it is not possible to write an output trace and then read the written value from it later.

Often, we want to simply copy the value from some address in the input traces to some address in the output trace.
While this is possible via a $\mathtt{@read}$ followed by a $\mathtt{@write}$, the language provides a special syntax:
\noindent
\begin{center}
\begin{tabular}{c}
{\begin{lstlisting}[basicstyle=\small\ttfamily]
@copy(<source-address>, <destination-address>)
\end{lstlisting}}
\end{tabular}
\end{center}
For example, to copy the value from address $\mathtt{:}\mathtt{u}$ in $x$ to address $\mathtt{:}\mathtt{v}$ in $x'$, we use:
\noindent
\begin{center}
\begin{tabular}{c}
{\begin{lstlisting}[basicstyle=\small\ttfamily]
@copy(model_in[:u], model_out[:v])
\end{lstlisting}}
\end{tabular}
\end{center}
Of course, it is also possible to copy from $x$ to $y'$, from $y$ to $x'$ and from $y$ to $y'$.
As we will see in Section~\ref{sec:sparsity}, it is preferable to use $\mathtt{@copy}$ when possible instead of reading and then writing, as this can make the acceptance probability calculation more efficient.

\paragraph{Constructing an Involution}
Consider the following generative model, which posits that a vector of univariate data is either generated from a single normal distribution or a mixture of two normal distributions.
The model is expressed as a probabilistic program:
\noindent
\begin{center}
\begin{tabular}{c}
\begin{lrbox}{\mybox}%
\begin{lstlisting}[basicstyle=\small\ttfamily]
@gen function p()
  k ~ uniform_discrete(1, 2)
  means = [{(:mu, j)} ~ normal(0, 10) for j in 1:k]
  weights = ones(k)/k
  vars = ones(k)
  for i in 1:100
    {(:x, i)} ~ mixture_of_normals(weights, means, vars)
  end
end
\end{lstlisting}
\end{lrbox}%
\scalebox{0.90}{\usebox{\mybox}}
\end{tabular}
\end{center}
This is a simplified version of the infinite Gaussian mixture model in Figure~\ref{fig:mixture}b.
Each key of the form $(\mathtt{:}\mathtt{x}, i)$ will be observed (in the dictionary $b$), so the latent part of the trace (the dictionary $x$) contains the other keys $\mathtt{:}\mathtt{k}$ and $(\mathtt{:}\mathtt{mu}, 1)$ and (when $\mathtt{k} = 2$) $(\mathtt{:}\mathtt{mu}, 2)$.

We will now walk through how to express a simple `split-merge' move~\citep{richardson1997bayesian} using our differentiable programming language.
Suppose we want to construct an involutive MCMC kernel that changes $\mathtt{k}$ from $2$ to $1$ or vice versa.
The discrete part of the involution is straightforward:
\noindent
\begin{center}
\begin{tabular}{c}
\begin{lrbox}{\mybox}%
\begin{lstlisting}[basicstyle=\small\ttfamily]
k = @read(model_in["k"], discrete)
if k == 1
    @write(model_out["k"], 2, discrete)
else
    @write(model_out["k"], 1, discrete)
end
\end{lstlisting}
\end{lrbox}%
\scalebox{1.0}{\usebox{\mybox}}
\end{tabular}
\end{center}
Writing the involution code for the continuous choices is more complex.
There is no one-to-one correspondence between the space of dictionaries with keys $\{\mathtt{:}\mathtt{k}, (\mathtt{:}\mathtt{mu}, 1)\}$ and the space of dictionaries $x$ with keys $\{\mathtt{:}\mathtt{k}, (\mathtt{:}\mathtt{mu}, 1), (\mathtt{:}\mathtt{mu}, 2)\}$.
Therefore, we need to extend the state space using the auxiliary distribution, defined with the following probabilistic program:
\noindent
\begin{center}
\begin{tabular}{c}
\begin{lrbox}{\mybox}%
\begin{lstlisting}[basicstyle=\small\ttfamily]
@gen function q(model_trace)
    if model_trace["k"] == 1
        # we are doing a split, sample extra DoF
        u ~ beta(2, 2)
    end
end
\end{lstlisting}
\end{lrbox}%
\scalebox{1.00}{\usebox{\mybox}}
\end{tabular}
\end{center}
Our involution is on the space of combined dictionaries $x \oplus y$ where $y$ are traces of $\mathtt{q}$.
Note that for each $(x \oplus y)$ where $x[\mathtt{k}] = 1$, $y$ has a key $\mathtt{u}$, and for each $(x \oplus y)$ where $x[\mathtt{k}] = 2$, $y$ is empty.
We extend the involution using a pair of bijections between $(\mu_1, \mu_2)$ and $(\mu_1, u)$ that show how two cluster means should be transformed into one cluster mean (merge), and vice versa (split):
\begin{equation} \label{eq:split-merge}
\begin{array}{c}
\mu_1, \mu_2 \mapsto ((\mu_1 + \mu_2) / 2, \mu_2 - (\mu_1 + \mu_2)/2) \;\; [\mathrm{Merge}]\\
\mu_1, u \mapsto (\mu_1 - u, \mu_1 + u) \;\; [\mathrm{Split}]
\end{array}
\end{equation}
The value of $x[\mathtt{k}]$ determines which of these functions is executed.
The full involution program $\mathcal{F}$ is then:
\noindent
\begin{center}
\begin{tabular}{c}
\begin{lrbox}{\mybox}%
\begin{lstlisting}[basicstyle=\small\ttfamily]
@transform f (model_in,aux_in) to (model_out,aux_out)
begin
    k = @read(model_in["k"], discrete)
    if k == 1
        @write(model_out["k"], 2, discrete)

        # split bijection
        mu = @read(model_in[("mu", 1)], continuous)
        u = @read(aux_in["u"], continuous)
        mu1 = mu - u; mu2 = mu + u
        @write(model_out[("mu", 1)], mu1, continuous)
        @write(model_out[("mu", 2)], mu2, continuous)
    else
        @write(model_out["k"], 1, discrete)

        # merge bijection
        mu1 = @read(model_in[("mu", 1)], continuous)
        mu2 = @read(model_in[("mu", 2)], continuous)
        mu = (mu2 + mu1) / 2
        u = mu2 - mu
        @write(model_out[("mu", 1)], mu, continuous)
        @write(aux_out["u"], u, continuous)
    end
end
\end{lstlisting}
\end{lrbox}%
\scalebox{0.90}{\usebox{\mybox}}
\end{tabular}
\end{center}
The two continuous bijections (implemented in blocks of code in the two branches) are inverses of one another.
This is a common pattern in involution programs---an involution on the discrete parts of the traces (in this case, just $x[\mathtt{k}]$) determines via control flow which continuous code blocks get executed, such that the end-to-end program defines an involution.

Note that the observations $b$ (in this case, the x-coordinates) are included in traces of $\mathtt{p}$, which take the form $x \oplus b$ where $x$ is the latent part and $b$ is the observed part.
The involution program $\mathcal{F}$ is allowed to read from the the observations in the model trace using the same syntax as used to read from the latent part of the trace $x$.
(The example given here does not utilize this feature).

\subsubsection{Computing the Jacobian with Automatic Differentiation}

Recall that the involution $f$ must decompose into (i) an involution $g$ on elements $e \in E$ of a partition of the state space, and (ii) a pair of continuous differentiable bijections $h_{e}$ and $h_{g(e)} = h_{e}^{-1}$ between each pair of corresponding elements of the partition.
Each function $h_e$ is a function from the values at continuous addresses in the input trace ($z[k]$ for $k \in K_z \cap \mathcal{J}$) to the values at continuous addresses in the output trace ($z'[k]$ for $k \in K_{z'} \cap \mathcal{J}$).
In the example above, the partition $E$ is given by the equivalence classes of the relation $(x_1, y_1) \sim (x_2, y_2) \iff x_1[\mathtt{k}] = x_2[\mathtt{k}]$ (there are two equivalence classes), $g$ maps the $k = 1$ class to the $k = 2$ class and vice versa, and 
the continuous bijections $h_e$ and $h_{g(e)}$ are the functions in Equation~(\ref{eq:split-merge}).

We compute the Jacobian using automatic differentiation. (The Gen implementation uses forward-mode AD whereas our PyTorch implementation uses reverse-mode AD.)
The procedure \textproc{run-involution} in Algorithm~\ref{alg:auto-involutive-mcmc} shows the implementation of the interpreter for the language that uses reverse-mode AD.
The interpreter executes $\mathcal{F}$ using regular Julia or Python semantics, but intercepts calls to \texttt{@write}, \texttt{@read}, and \texttt{@copy} statements, and in addition to performing the desired operation, records the set of continuous addresses that are read, written, and copied.
After $\mathcal{F}$ is finished executing, AD uses the recorded addresses to compute the Jacobian $|J h_e|(z)$; in the case of reverse-mode AD, this is accomplished by iterating over output continuous addresses and backpropagating from each one to all of the input continuous addresses, computing the Jacobian column-by-column.
Keys with an $n$-dimensional Lebesgue measure as their reference measure, corresponding to vector-valued random choices, are unpacked into $n$ separate columns in the Jacobian; this detail is elided in Algorithm~\ref{alg:auto-involutive-mcmc}.
The $\mathtt{discrete}$ and $\mathtt{continuous}$ labels are also omitted from the syntax in \textproc{run-involution} to simplify notation ($a \in \mathcal{I}$ indicates a discrete choice and otherwise a choice is continuous).

\section{EXPLOITING SPARSITY FOR IMPROVED PERFORMANCE} \label{sec:sparsity}

Suppose $N$ and $N'$ are the total number of random choices in the input traces $x, y$ and output traces $x', y'$ respectively
(note that the observations $b$ are excluded---$x$ and $x'$ are the \emph{latent} part of the model's traces only).
Let $M$ be the number of continuous random choices among these (which must be the same in the input traces and output traces).
Then, the number of operations used in Algorithm~\ref{alg:auto-involutive-mcmc} grows as $O(N + N') + O(M^3)$.
The linear term is due to sampling $y'$ and computing the four log-densities required for the acceptance probability.
The cubic term is due to computing the Jacobian determinant, which is also required for the acceptance probability (Equation~\ref{eq:acceptance-ratio-dictionaries}).

It is possible to reduce the number of operations performed in an automated involutive MCMC kernel by exploiting special structure in the involution $f$.
In some cases, this structure can lead to $O(1)$ operations per kernel application (i.e. constant in $N$, $N'$, and $M$).
This section describes techniques for exploiting involution structure within an automated involutive MCMC implementation.
These techniques are used in the Gen implementation, and one of the techniques is used in our minimal PyTorch implementation.

\subsection{Sparsity-Aware Automatic Jacobian Computation}

The naive implementation of Algorithm~\ref{alg:auto-involutive-mcmc} computes by the Jacobian by first computing the $M$-by-$M$ Jacobian matrix $J h_{e(z)}$ via automatic differentiation, and then computing the absolute value of its determinant.
However, we observe that in many applications of involutive MCMC, the values at continuous choices in the input traces are directly copied into the output traces (either at the same key or a different key).
These copy operations result in columns in the Jacobian matrix that have a single $1$ entry with remaining entries equal to $0$.
For example, for the function $(u, v, x, y) \mapsto (u, 2 u - v, y, x) = (u', v', x', y')$, the Jacobian is (with columns corresponding to $u'$, $v'$, $x'$, and $y'$ and rows corresponding to $u$, $v$, $x$, and $y$):
\begin{equation*}
\left[
\begin{array}{cccc}
1 & 2 & 0 & 0\\
0 & -1 & 0 & 0\\
0 & 0 & 0 & 1\\
0 & 0 & 1 & 0
\end{array}
\right]
\begin{array}{c}
u\\ v\\ x\\ y
\end{array}
\end{equation*}
Using the cofactor expansion of the determinant, we observe that for any `copy' column in an $M$-by-$M$ Jacobian matrix (a column with a single $1$ and all other entries $0$), the absolute value of the determinant is equivalent to that of the $(M-1)$-by-$(M-1)$ sub-matrix with the corresponding column \textit{and} row omitted (even if that would remove other nonzero entries from the matrix).
By applying this rule recursively, we can instead compute the determinant of a much smaller matrix; for the example above with $M=4$, the absolute value of the determinant simplifies to the absolute value of a single entry ($|-1|$).
Indeed, if some input key is copied to some output key, then we can entirely avoid computing its row (and corresponding column) of the Jacobian.
Therefore, the number of operations (which is dominated by the determinant) reduces from $M^3$ to $(M-|C|)^3$ where $|C|$ is the number of input keys that were copied to some output key.
The $\mathtt{@copy}$ statement in our differentiable programming language makes the set of input keys that were copied explicitly available to the interpreter, which makes automating this optimization straightforward, as shown in \textproc{run-involution} in Algorithm~\ref{alg:auto-involutive-mcmc}.

Many involutive MCMC kernels modify only a constant number of keys that does not depend on the sizes $N$ and $N'$ of the input and output traces or the total number of continous keys $M$; the other keys are copied over unchanged.
For example, consider the Jacobian for a split-merge reversible jump move~\citet{richardson1997bayesian}, which is implemented in Figure~\ref{fig:mixture}:
\includegraphics[width=\linewidth]{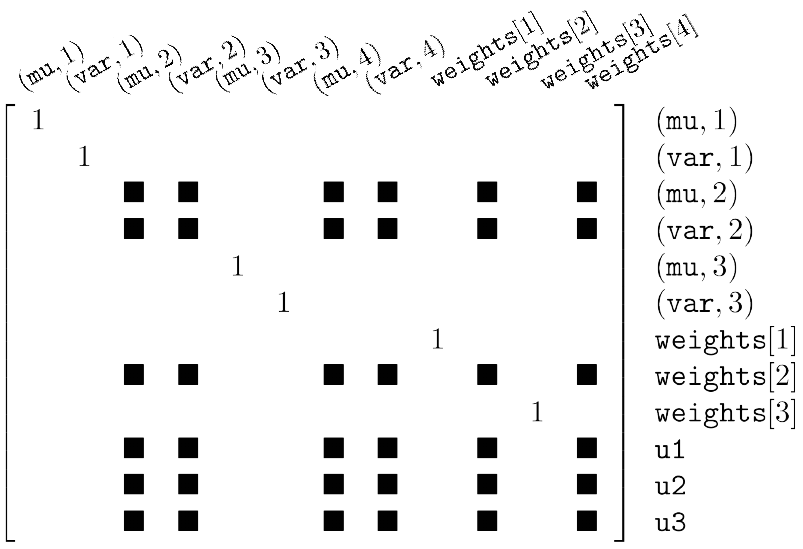}
Here, the move is splitting cluster 2 into two clusters (cluster 2 and cluster 4).
Black squares indicate nonzero entries.
In this case, the size $N$ of the latent part of the trace grows linearly in the number of clusters, but the Jacobian determinant can be calculated from only the $6$-by-$6$ submatrix of the Jacobian that involves the parameters of the one cluster being split (or the two clusters being merged).
Therefore, the acceptance ratio computation reduces from $O(N + N') + O(M^3)$ to $O(N + N')$.
While a reasonable hand-coded implementation of this algorithm would likely perform this sort of optimization as well, Algorithm~\ref{alg:auto-involutive-mcmc} automates it.

\subsection{Incremental Computation of Output Traces and Density Ratios}

Consider input traces $(x, y)$ and output traces $(x', y')$ with $f(x, y) = (x', y')$ for involution $f$ encoded by program $\mathcal{F}$ (again, $x$ and $x'$ include only the \emph{latent} part of the model trace and not the observed part $b$).
Let $N_1 := |K_{x}|$, $N_2 := |K_{y}|$, $N_1' := |K_{x'}|$ and $N_2' := |K_{y'}|$, and
$N := N_1 + N_2$ and $N' := N_1' + N_2'$.
Suppose that $N_1 + N_1' \gg N_2 + N_2'$, which often occurs when the the number of latent variables is large, but the involutive MCMC move only updates a portion of the latent variables.
Algorithm~\ref{alg:auto-involutive-mcmc} runs $\mathcal{F}$, which explicitly writes the value for each element of $(x', y')$, requiring $N' = N_1' + N_2' \approx N_1'$ `write' or `copy' operations.
The algorithm also uses approximately $N_1 + N_1'$ operations to evaluate the log-densities, because it accumulates the log-density of each random choice in $x$, and $x'$.

We now show how to modify the program $\mathcal{F}$ and its interpreter, so that both computing $z' = (x', y')$ from $z = (x, y)$, and computing the log density ratio $\log(\tilde{p}(x' \oplus b) / \tilde{p}(x \oplus b))$, can take $O(1)$ operations (that is, constant in $N_1 + N_1'$) when the involution $f$ has sparse structure and the model density $p$ has conditional independencies.

Note that $x$ and $x'$ may both have values for certain keys, and $x$ may contains keys not present in $x'$, and vice versa.
Often, an involution does not modify the values of many keys in the trace---this is the case in the split-merge reversible jump move described earlier.
Suppose that $A \subseteq K_{x} \cap K_{x'}$ is the set of keys $k$ in both $x$ and $x'$ for which $x'[k] = x[k]$.
Let $\Delta := K_{x'} \setminus A$ be the set of keys in $x'$ that are either (i) not in $x$, or (ii) in $x$ but have their value changed.
Then, the previous trace $x$ and a dictionary $\delta$ with $K_{\delta} = \Delta$ and $\delta[k] = x'[k]$ and the density function $p$ suffice to uniquely define the trace $x'$, provided the density satisfies the following condition, which is satisfied by densities defined by probabilistic programs\footnote{Intuitively, two traces of a probabilistic program cannot have different sets of addresses unless they share an address and disagree on its value.}:
$p(x) > 0$ and $p(x') > 0$ and $x \ne x'$ implies there exists $k \in K_x \cap K_{x'}$ with $x[k] \ne x'[k]$.
Let $x|_A$ denote the \emph{restriction} of a dictionary $x$ to only keys in $A$.
Then, \textproc{naive-trace-update} (Algorithm~\ref{alg:naive-update}) computes the new trace $x'$ and log-densities ratio $\log (p(x') / p(x))$ of a probabilistic program from a previous trace $x$ and a dictionary $\delta$ (and observations $b$), provided there exists $x' = \delta \oplus x|_A$ for some $A$ where $p(x') > 0$.
\begin{algorithm}[h]
\begin{algorithmic}
\Require probabilistic program $\mathcal{P}$; previous trace $x$; observations $b$;\\
dictionary $\delta$ such that there exists $x' = (\delta \oplus x|_{A})$ for some $A$ where $p(x') > 0$.
\vspace{1mm}
\Procedure{naive-trace-update}{$\mathcal{P}$, $x$, $b$, $\delta$}
   \State {\bf assert} $K_{x} \cap K_{b} = \varnothing$
   \State $x' \gets \{\}$
   \State $s' \gets 0$
   \execute[$\mathcal{P}$], but with $\mathtt{"}k \sim \mathrm{distribution}\mathtt{"} \equiv$ (
        \State 1. if $k \in K_{\delta}$ then $v \gets \delta[k]$ else $v \gets (x \oplus b)[k]$
        \State 2. if $k \not \in K_{b}$ then set $x'[k] \gets v$
        \State 3. set $s' \gets s' + \textproc{logpdf}(\mathrm{distribution}, v)$
        \State 4. evaluate to $v$)
    \END
   \State $s \gets \textproc{score}(\mathcal{P}, x \oplus b)$
   \State \Return $(x', s' - s)$
\EndProcedure
\end{algorithmic}
\caption{Naive Trace Update}
\label{alg:naive-update}
\end{algorithm}

\textproc{naive-trace-update} is a naive implementation of a general \emph{trace-update} operation that is part of Gen's API~\citep{cusumano2019gen}.
While the naive implementation requires $O(N_1 + N_1')$ operations, it is possible for more sophisticated implementations of \emph{trace-update} to run in approximately $O(|K_{\delta}|)$ operations for certain probabilistic programs $\mathcal{P}$ by exploiting conditional independence structure program that leads to cancellations in the density ratio $\tilde{p}(x' \oplus b) / \tilde{p}(x \oplus b)$.
The details of these more efficient implementations are outside the scope of this paper, but are implemented as part of the Gen system.

Given an implementation of \emph{trace-update} (that has the same input and output signature of \textproc{naive-trace-update}) we can optimize Algorithm~\ref{alg:auto-involutive-mcmc} by
modifying the requirement for the involution program $\mathcal{F}$:
Instead of explicitly specifying a value for all keys in $x'$, it need only explicitly specify values in $\delta$.
(We call the other values in $x'$ \emph{implicit copies}---they could be explicitly copied within $\mathcal{F}$, but this would result in uncessary code and possibly unecessary computation.)
Therefore, the involution program can run in $O(|K_{\delta}| + N_2 + N_2')$ operations.
We remove the separate calls to \textproc{score} for the log-densities $\log \tilde{p}(x \oplus b)$ and $\log \tilde{p}(x' \oplus b)$ and replace them with a single call to \emph{trace-update} on $\mathcal{P}$ that takes in $x$ and $\delta$ and returns $x'$ and $\log (\tilde{p}(x' \oplus b) / \tilde{p}(x \oplus b))$.
When $|K_{\delta}|$ and $N_2$ and $N_2'$ are all constant in the size $N_1$ of the latent part of model trace ($x$), the resulting algorithm uses a number of operations that does not grow with $N_1 + N_1'$.

\section{DYNAMICALLY DETECTING BUGS IN INVOLUTIONS}

The automated technique described in the previous sections allows users to implement involutive MCMC algorithms by writing probabilistic programs for $p$ and $q$ and a differentiable program for the involution $f$, avoiding the need to derive and implement the accept/reject formula by hand. 
This eliminates certain classes of errors that could otherwise be difficult to root out in hand-coded implementations of MCMC algorithms. However, it is still of course possible to implement $p$, $q$, or $f$ incorrectly, introducing bugs that may invalidate the correctness proofs for the involutive MCMC algorithm.

Fortunately, the factorization of diverse algorithms into a common template, involving an explicitly represented model, auxiliary distribution, and involution, enables simple and automated debugging checks that can catch a wide variety of such errors dynamically. These checks are not complete, but anecdotally, we have found that they often alert users to subtle bugs in the design or implementation of an algorithm. The checks are summarized in Algorithm~\ref{alg:dynamic-check}.

\begin{algorithm}[h]
\begin{algorithmic}
\Procedure{involutive-mcmc-check}{$\mathcal{P}$, $\mathcal{Q}$, $\mathcal{F}$}
    \State \vspace{-2mm}
    \LineComment{Generate a randomized test case}
    \State $((x \oplus b), \_) \sim \textproc{trace-and-score}(\mathcal{P})$ 
    \State $(y, \_) \sim \textproc{trace-and-score}(\mathcal{Q}_x)$
    \State \vspace{-2mm}
    \LineComment{Run involution and perform dimension check}
    \State $x' \oplus y' \gets \textproc{run-involution}(\mathcal{F}, x \oplus y, b)$
    \State \vspace{-2mm}
    \LineComment{Support check}
    \State $\log p(x') \gets \textproc{score}(\mathcal{P}, x')$
    \State $\log q_{x'}(y') \gets \textproc{score}(\mathcal{Q}_{x'}, y')$
    \State {\bf assert} $(\log p(x') > -\infty) \land (\log q_{x'}(y') > -\infty)$
    \State \vspace{-2mm}
    \LineComment{Involution check}
    \State $\tilde{x} \oplus \tilde{y} \gets \textproc{run-involution}(\mathcal{F}, x' \oplus y', b)$
    \State ${\bf assert} (\tilde{x} = x) \land (\tilde{y} = y)$
\EndProcedure
\end{algorithmic}
\caption{Dynamic Check for Bugs}
\label{alg:dynamic-check}
\end{algorithm}

The check first randomly samples a model trace $x$ and simulated observations $b$ from the model prior $p$, and auxiliary trace $y$ from the $q_x$. It then runs three tests on this simulated test case:

\begin{enumerate}
\item \textbf{Support check.} The \textit{support check} runs the involution $\mathcal{F}$ on $(x, y)$ and checks that the resulting pair of traces, $(x', y')$, are within the set $Z$ of positive-density elements for $\pi$.

\item \textbf{Dimension check.} The \textit{dimension check} runs the involution $\mathcal{F}$ on $(x, y)$ and records the sets of addresses $R$ and $W$ at which $\mathcal{F}$ performs continuous reads and writes, respectively, and checks that the sizes of these sets match.
(This check is part of computing the determinant of the Jacobian in \textproc{run-involution}).

\item \textbf{Involution check.} The \textit{involution check} runs $\mathcal{F}$ twice to compute $f(f(x, y))$ and checks that the result is equal to $(x, y)$. If $f$ is an involution on $Z$, then this check will always succeed; if there is a set $\tilde{Z}$ of positive $\pi$-measure on which $f(f(z)) \neq z$ (i.e., $f$ is not an involution), then this check has positive probability of failing.
\end{enumerate}

These checks each catch qualitatively different bugs in user programs. We now give several examples.

\textit{Incorrectly transformed continuous variables.} Many bugs in the design or implementation of deterministic transformations of continuous variables are naturally detected by the involution check. Consider, for example, the Hamiltonian Monte Carlo algorithm~\citep{hmc}, which, as \citet{imcmc} observe, is an instance of the involutive MCMC framework. 
When applied to HMC, the dynamic bug check involves:
(a) sampling a model state $x$ from the prior; (b) sampling a momentum $y$; (c) running the leapfrog integrator forward to get a new state $x'$ and new momentum $-y'$, (d) running the leapfrog integrator backward from state $x'$ with momentum $y'$, and (e) checking that this results in the state-momentum pair $x, -y$. If the momentum is not properly negated, or if the leapfrog integrator is incorrectly implemented, this check can fail.

As another, simpler example of this type of error, consider the \texttt{split\_params} and \texttt{merge\_params} functions invoked in Figure~\ref{fig:mixture}, as part of a reversible-jump MCMC kernel for inferring a mixture of Gaussians with an unknown number of components. Considering only the $\mu$ parameter, suppose the \textit{split} move samples an auxiliary variable $u_1$ and computes $\mu_1 = \mu + u_1, \mu_2 = \mu - u_1$ as the means of the two new clusters. If the \textit{merge} move joins two clusters and assigns $\mu = \sqrt{\mu_1 \mu_2}$, there is a mismatch: a split move cannot be reversed by a corresponding merge, because in general, $\sqrt{(\mu+u_1)(\mu-u_1)} \neq \mu$. This error could be fixed either by changing the split to compute $\mu_1 = \mu u_1, \mu_2 = \frac{\mu}{u_1}$, or by changing the merge to compute $\mu = \frac{\mu_1 + \mu_2}{2}$.

%


\textit{Discrete logic errors.} Another common class of errors is 
for the discrete logic of an involution to be flawed. Consider the following incorrect 
implementation of a \textit{birth-death} move for the mixture model in Figure~\ref{fig:mixture}, which either adds a new mixture component or selects an existing one
at random to delete
(we name the traces $\mathtt{tr1}$, $\mathtt{tr2}$, $\mathtt{tr3}$, ad $\mathtt{tr4}$ for reasons of space):

\begin{center}
\begin{tabular}{c}
{\begin{lstlisting}[basicstyle=\small\ttfamily]
@gen function q(trace)
  current_k = trace[:k]
  is_birth ~ bernoulli(current_k == 1 ? 1.0 : 0.5)
  if is_birth
    new_mu ~ normal(0, 10)
    new_var ~ inv_gamma(1, 10)
  else
    deletion_idx ~ uniform_discrete(1, current_k)
  end
end
\end{lstlisting}}
\end{tabular}
\end{center}

\begin{center}
\begin{tabular}{c}
{\begin{lstlisting}[basicstyle=\small\ttfamily]
@transform h (tr1, tr2) to (tr3, tr4) begin
   is_birth = @read(tr2, :is_birth, :discrete)
   @write(tr4[:is_birth], !is_birth, :discrete)
   k = @read(tr1[:k], :discrete)
   weights = @read(tr1[:weights], :continuous)
   if is_birth
     @write(tr3[:k], k+1, :discrete)
     new_mu = @read(tr2[:new_mu], :continuous)
     new_var = @read(tr2[:new_var], :continuous)
     @write(tr3[(:mu, k+1)], new_mu, :continuous)
     @write(tr3[(:var, k+1)], new_var, :continuous)
     new_weights = add_weight(weights)
     @write(tr4[:deletion_idx], k+1, :discrete)
   else
     idx = @read(tr2[:deletion_idx], :discrete)
     @copy(tr1[(:mu, idx)], tr4[:new_mu])
     @copy(tr1[(:var, idx)], tr4[:new_var])
     for i in (idx+1):k
       @copy(tr1[(:mu, i)], tr3[(:mu, i-1)])
       @copy(tr1[(:var, i)], tr3[(:var, i-1)])
     end
     @write(tr3[:k], k-1, :discrete)
     new_weights = delete_weight(weights, idx)
   end
   @write(tr3[:weights], new_weights, :continuous)
end
\end{lstlisting}}
\end{tabular}
\end{center}

The flaw in this implementation is that although the death move can 
delete any of the $k$ mixture components, the birth move can only add
a new component to the \textit{end} (index $k+1$), so the move is not
reversible. The involution check will discover that the deletion of a
component with index $i < k$ is not reversed by a corresponding birth move,
and will thus raise an error.

\textit{Other miscellaneous errors.} When implementing distributions as probabilistic programs, it is also possible for users to make more mundane errors, such as spelling the name of a random choice inconsistently, or characterizing random choices using the wrong \texttt{type} tags (\texttt{:continuous} and \texttt{:discrete}). Such errors can be difficult to detect statically, because the addresses at which a probabilistic program makes random choices, and the distributions of those choices, may change from sample to sample. (Previous work has explored static analyses based on types~\citep{popltraces} and abstract interpretation~\citep{verified_svi}, but these each work on limited subsets of the programs that Gen's full modeling language permits, and it is often precisely these more complex programs that require the flexibility of the involutive MCMC framework in the first place). Our dynamic support and dimension checks can help to detect bugs like these. For example, if the involution $f$ writes to a misspelled address, the support check will determine that the resulting trace's density is 0.


\textbf{Dynamic checks during inference.} These dynamic assertions can also be run during inference, at each application of the transition kernel. This can be useful to catch bugs that only occur in regions of the state space with low prior mass (but perhaps high posterior mass). When enabled, we can run dynamic checks after each application of the kernel, and when they fail, write to a debugging log and reject the proposed new state. As it turns out, the kernel induced by this procedure is still stationary for $p$:

\begin{lemma}
Let $p$ and $q_x$ be model and auxiliary densities as above, but suppose $f : \mathcal{D} \times \mathcal{D} \rightarrow \mathcal{D} \times \mathcal{D}$ may not be an involution on $Z$.  Trace-based involutive MCMC with dynamic checks enabled, rejecting whenever such a check fails, still yields a kernel that is stationary for $p$.
\end{lemma}
\begin{proof}
Let $R = \{x \in Z \mid f(f(x)) = x \wedge \pi(f(x)) > 0\}$, and let $h(x) := \mathbf{1}[x \in R]f(x) + \mathbf{1}[x \not\in R]x$. Then $h$ is an involution on $Z$, and trace-based involutive MCMC with $p$, $q_x$, and $h$ yields a stationary kernel. But this kernel is the same one induced by using $f$ with dynamic checks. For $x$ on which dynamic checks succeed, $h$ is equivalent to $f$. For $x$ on which dynamic checks fail, $h$ is equivalent to the identity; thus, accepting a move produced by $h$ is equivalent to rejecting.  
\end{proof}

\section{EXAMPLES}

\subsection{Reversible Jump MCMC}
Reversible jump MCMC~\citep{green1995reversible,hastie2012model} is a special case of involutive MCMC, and 
the implementation of reversible jump MCMC kernels can be automated using the probabilistic and differentiable programming languages presented in this paper.
We now review reversible jump MCMC, then show how it can be automated using the techniques presented earlier, and give an example.

\paragraph{Review of reversible jump MCMC.}
The reversible jump MCMC framework involves a set of `models' $h \in \mathcal{H}$, and a prior distribution on models $p(h)$.
For each model, there is a latent continuous parameter vector $\theta_h \in \mathbb{R}^{n(h)}$ where $n(h)$ is the dimension of model $h$, and a likelihood function $L_{D,h}(\theta_h)$ for each $h$ given data $D$.
The latent state $x$ is a pair $(h, \theta_h)$ of model and continuous parameter.
There is a set of \emph{move types} $\mathcal{M}$.
Each move type $m \in \mathcal{M}$ is associated with an unordered pair of models $(h_1, h_2)$ and a dimensionality $d(m)$ such that $d(m) \ge n(h_1)$ and $d(m) \ge n(h_2)$ (zero, one, or more than one move types may be associated with a given pair of models).
For each latent state $x = (h, \theta_h)$, there is a probability distribution $q_{x}(m)$ on move types such that $q_{x}(m) > 0$ implies that $h$ is one of the models for move type $m$.
For each move type $m \in \mathcal{M}$ between $h_1$ and $h_2$ there is a pair of continuously differentiable bijections $g_{m, h_1 \to h_2} : \mathbb{R}^{d(m)} \to \mathbb{R}^{d(m)}$ and $g_{m, h_2 \to h_1} := g_{m, h_1 \to h_2}^{-1}$, and a pair of proposal densities $q_{m,h_1 \to h_2}(u_{h_1 \to h_2})$ and $q_{m,h_2 \to h_1}(u_{h_2 \to h_1})$ where $u_{h_1 \to h_2} \in \mathbb{R}^{d(m) - n(h_1)}$ and $u_{h_2 \to h_1} \in \mathbb{R}^{d(m) - n(h_2)}$.
A proposal is made from state $x = (h, \theta_h)$ by (i) sampling a move type $m \sim q_{x}(\cdot)$, and (ii) sampling continuous variable $u \sim q_{m,h \to h'}(\cdot)$ for $(h, h')$ associated with $m$, and (iii) computing $(\theta_{h}', u') := g_{m,h \to h'}(\theta_h, u)$, and proposing new state $x' = (h', \theta_h')$.

\paragraph{Encoding reversible jump in involutive MCMC.}
To encode reversible jump MCMC in our framework, we write a probabilistic program $\mathcal{P}$ that encodes the space of models $\mathcal{H}$, the prior distribution on models, $p(h)$, the per-model priors $p_h(\theta_h)$ and the per-model likelihoods $L_{D,h}(\theta_h)$.
The set of all models $h$ is encoded in the set of all pairs $(K_\mathcal{P}, \mathbf{d}_\mathcal{P})$ where $K_\mathcal{P}$ represent possible trace structures (i.e. control-flow paths through $\mathcal{P}$) and $\mathbf{d}_\mathcal{P}$ are the set of assignments to discrete random choices made by $\mathcal{P}$.
The per-model continuous parameters $\theta$ are encoded via continuous random choices $\mathbf{c}_\mathcal{P}$.
The auxiliary probabilistic program $\mathcal{Q}$ encodes both the probability distribution on moves types using discrete random choices and possibly stochastic control flow ($(K_\mathcal{Q}, \mathbf{d}_\mathcal{Q})$), and the per-move-type probability densities on $u$ using continuous random choices $\mathbf{c}_\mathcal{Q}$.
The involution $f$ factors into an (i) involution $f_1$ on pairs $i = ((K_\mathcal{P}, \mathbf{d}_\mathcal{P}), (H_\mathcal{Q}, \mathbf{d}_\mathcal{Q}))$ that defines the association between move types $\mathcal{M}$ and the model pairs ($h_1, h_2$); and (ii) a family of bijections $f_{1,i}$ on the space of pairs $(\mathbf{c}_\mathcal{P}, \mathbf{c}_\mathcal{Q})$ of continuous random choices for both programs for fixed values of the discrete random choices and fixed trace structure.

\paragraph{Example: Split-merge reversible jump.}
Figure~\ref{fig:mixture} shows a split-merge reversible jump kernel for an infinite Gaussian mixture model~\citep{richardson1997bayesian} implemented using the probabilistic and differentiable programming languages described in this paper.
Figure~\ref{fig:mixture}b shows the infinite Gaussian mixture model, specified as a probabilistic program $\mathtt{p}$.
The program takes the number of data points as input, then samples the number of clusters from a Poisson distribution, then samples cluster parameters and mixture proportions, and finally samples the data points from the resulting finite mixture.
Figure~\ref{fig:mixture}c shows the auxiliary probabilistic program $\mathtt{q}$ for the split-merge kernel.
This program takes a trace of the model program as input, and randomly decides whether to split a cluster and increase the number of clusters by one or merge two clusters and decrease the number of clusters by one.
Then, the program randomly picks which cluster to split, or which clusters to merge.
This kernel always merges the last cluster with a random other cluster; for ergodicity the move can be composed with a simple move (that has acceptance probability $1$) that swaps a random cluster with the last cluster.
If a split is chosen, then the program also samples the three degrees of freedom necessary to generate the new parameters for the clusters in an invertible manner.
Figure~\ref{fig:mixture}d shows a differentiable program specifying the involution for the split-merge kernel, and Figure~\ref{fig:mixture}f shows graphically how this involution acts on pairs of traces.
The yellow section (1) defines an involution $f_1$ on the discrete random choices that specifies that (i) the \texttt{split} choice should be flipped (so that split moves are always mapped to merge moves and vice versa) and that (ii) the number of clusters should be increased by one for a split move and decreased by one for a merge move, and (iii) which merged cluster corresponds to which split clusters.
The green section (2) specifies the continuous bijections that govern the transformation of continuous random choices during split moves and the purple section (3) specifies the inverses of these bijections, which govern the transformation of continuous choices during merge moves.

\subsection{State-Dependent Mixture Proposals}

\paragraph{Example: Bayesian structure learning for Gaussian processes}
Figure~\ref{fig:structure-learning} shows automated involutive MCMC being applied to fully Bayesian inference over the covariance function of a Gaussian process, where the prior on covariance functions (Figure~\ref{fig:structure-learning}a) is based on a probabilistic context-free grammar.
The inference algorithm is based on an involutive MCMC kernel that uses a state-dependent mixture of proposals.
A variant of this inference algorithm was previously studied in~\citep{schaechtle2016time,saad2019bayesian} based on a model of~\citet{grosse2012exploiting}.

\paragraph{Hierarchical address spaces.}
This example uses probabilistic programs that invoke other probabilistic programs, sometimes recursively.
For example, the model probabilistic program $\mathtt{p}$ invokes the probabilistic program $\mathtt{cov\_function\_prior}$, which is itself recursive.
Similarly, the auxiliary probabilistic program $\mathtt{q}$ invokes $\mathtt{walk\_tree}$ (which is recursive) as well as $\mathtt{cov\_function\_prior}$.
Consider the syntax used to recursively invoke $\mathtt{walk\_tree}$ within $\mathtt{walk\_tree}$:
\begin{center}
\begin{tabular}{c}
{\begin{lstlisting}[basicstyle=\small\ttfamily]
({:left} ~ walk_tree(node.left, path))
\end{lstlisting}}
\end{tabular}
\end{center}
This expression resembles a random choice expression.
However, instead of associating the return value of the function $\mathtt{walk\_tree}$ with the address $\mathtt{:}\mathtt{left}$, the address $\mathtt{:}\mathtt{left}$ is associated with the entire trace of random choices made within $\mathtt{walk\_tree}$.
That is, $\mathtt{:}\mathtt{left}$ is the \emph{namespace} for the addresses of all random choice made within the invocation.
Further invocations by the callee themselves result in nested namespaces.
This process results in a \emph{hierarchical address space} for random choices, as shown in Figure~\ref{fig:structure-learning}b.
This does not modify the mathematical formalism---each random choice made during the execution of a probabilistic program still has a unique address, but the address has multiple components that localize it within the hierarchy.
For example, the choice $\mathtt{done}$ a recursive call to $\mathtt{walk\_tree}$ might have address:
\[
k = (\mathtt{:}\mathtt{left} \mbox{ \texttt{=>} } \mathtt{:}\mathtt{right} \mbox{ \texttt{=>} } \mathtt{:}\mathtt{done})
\]
(`\texttt{=>}' is the Gen syntax for constructing hierarchical addresses).

\paragraph{A complex state-dependent distribution.}
At each iteration of the MCMC algorithm, the auxiliary probabilistic program $\mathtt{q}$ (Figure~\ref{fig:structure-learning}d) first picks a random node in the parse tree of the covariance function, by doing a stochastic walk of the existing parse tree that terminates at the chosen node.
\begin{center}
\begin{tabular}{c}
{\begin{lstlisting}[basicstyle=\small\ttfamily]
prev_cov_function = trace[:cov_function]
path ~ walk_tree(prev_cov_function, ..)
\end{lstlisting}}
\end{tabular}
\end{center}
The code that walks the tree uses the following recursion, which results in a probability distribution that assigns exponentially lower probability to nodes that are deeper in the tree.
\begin{center}
\begin{tabular}{c}
{\begin{lstlisting}[basicstyle=\small\ttfamily]
if ({:done} ~ bernoulli(0.5))
  return path
elseif ({:recurse_left} ~ bernoulli(0.5))
  path = (path..., :left_node)
  return ({:left} ~ walk_tree(node.left, path))
else
  path = (path..., :right_node)
  return ({:right} ~ walk_tree(node.right, path))
end
\end{lstlisting}}
\end{tabular}
\end{center}
The resulting distributions on selected nodes for two possible input trees are shown below:
\begin{center}
\includegraphics[width=0.8\linewidth]{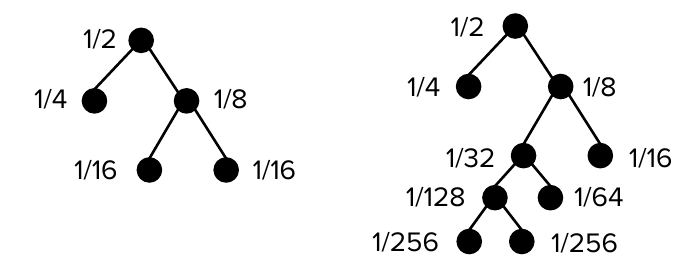} 
\end{center}
The first part of the involution (Figure~\ref{fig:structure-learning}e) copies the random choices made during this walk from the input auxiliary trace to the output auxiliary trace.
\begin{center}
\begin{tabular}{c}
{\begin{lstlisting}[basicstyle=\small\ttfamily]
@copy(aux_in[:path], aux_out[:path])
\end{lstlisting}}
\end{tabular}
\end{center}
Note that here, $\mathtt{@copy}$ is being used to copy the entire \emph{set} of random choices from the namespace $\mathtt{:}\mathtt{path}$ in $\mathtt{aux\_in}$ to the namespace $\mathtt{:}\mathtt{path}$ in $\mathtt{aux\_out}$.

Because the mixture distribution is specified using a probabilistic program, it is straightforward to modify the program $\mathtt{p}$ to define a different mixture distribution.
The code below specifies a mixture distribution that is uniform over all nodes in the tree.
\begin{center}
\begin{tabular}{c}
{\begin{lstlisting}[basicstyle=\small\ttfamily]
n1 = size(node.left); n2 = size(node.right)
if ({:done} ~ bernoulli(1 / (1 + n1 + n2)))
  return path
elseif ({:recurse_left} ~ bernoulli(n1 / (n1+n2))
  path = (path..., :left_node)
  return ({:left} ~ walk_tree(node.left, path))
else
  path = (path..., :right_node)
  return ({:right} ~ walk_tree(node.right, path))
end
\end{lstlisting}}
\end{tabular}
\end{center}
The resulting distributions, for two possible input trees, are:
\begin{center}
\includegraphics[width=0.8\linewidth]{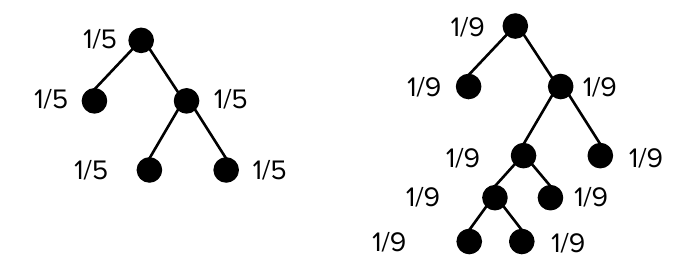}
\end{center}
Note that the probability of choosing a given subtree to propose to is itself changed when the subtree changes.
Therefore, the mixture probabilities do not in general cancel in the the acceptance probability calculation, and must be accounted for.
For the original mixture distribution, the ratio of mixture probabilities is either $1$, $0.5$, or $2$ depending on whether the previous and new subtrees are leaf or internal nodes.
For this alternative mixture distribution, the ratio of mixture probabilities is the ratio of sizes of the two trees (e.g. $9/5$ or $5/9$ for the trees above).
In both cases, our automated involutive MCMC algorithm automatically computes the acceptance probability.

\paragraph{A general pattern for state-dependent mixtures of proposals in Metropolis-Hastings}
The other parts of the auxiliary probabilistic program and the involution program specifies a proposal distribution for the subtree of the parse tree that is rooted at the chosen node.
In particular, the rest of the auxiliary probabilistic program $\mathtt{q}$ proposes a new subtree by sampling from the same process used to recursively define the prior distribution:
\begin{center}
\begin{tabular}{c}
{\begin{lstlisting}[basicstyle=\small\ttfamily]
new_subtree ~ cov_function_prior()
\end{lstlisting}}
\end{tabular}
\end{center}
The involution program swaps the old subtree with the newly proposed subtree:
\begin{center}
\begin{tabular}{c}
\begin{lrbox}{\mybox}%
\begin{lstlisting}[basicstyle=\small\ttfamily]
@copy(model_in[subtree_address], aux_out[:new_subtree])
@copy(aux_in[:new_subtree], model_out[subtree_address])
\end{lstlisting}
\end{lrbox}%
\scalebox{0.9}{\usebox{\mybox}}
\end{tabular}
\end{center}

This is an instance of a more general pattern for implementing state-dependent mixture proposals:
\begin{enumerate}
\item The auxiliary probabilistic program $\mathcal{Q}$ samples from a distribution over different sets of random choices that will be proposed to
(in this case, each set is a different subtree of the parse tree).
\item The auxiliary probabilistic program $\mathcal{Q}$ then samples new values for those random choices (in this case, a new subtree).
\item The involution program $\mathcal{F}$ swaps the previous values of those random choices with their new values, by swapping data between the model trace and the auxiliary trace.
\item The involution program $\mathcal{F}$ copies the random choices that determined what subset of random choices to propose to from the input auxiliary trace to the output auxiliary trace.
\end{enumerate}

\section{DISCUSSION}

Involutive MCMC is a unifying construction~\citep{imcmc} for MCMC algorithms that encompasses both classic approaches to constructing kernels like reversible jump MCMC~\citep{green1995reversible}, but also recently introduced classes of MCMC kernels based on neural networks~\citep{spanbauer2020deep}.
Therefore, the approach to automating involutive MCMC kernels presented in this paper makes a number of classic MCMC techniques easier to use and broadens their accessibility, and may potentially aid in development of novel MCMC techniques as well.
The implementation of our approach in the Gen probabilistic programming system has already been used by researchers in computational biology~\citep{merrell2020inferring} and artificial intelligence~\citep{zhi2020online} to prototype and develop new reversible-jump MCMC algorithms.

The technique for automating involutive MCMC presented in this paper can be generalized to the setting of sequential Monte Carlo samplers~\citep{del2006sequential}.
Instead of one model probabilistic program, one auxiliary probabilistic program and one involution program, there are two model probabilistic programs, two auxiliary probabilistic programs, and a pair of bijective differentiable programs that transform traces of one model into traces of the other.
This more general construct, which builds on earlier work on sequential Monte Carlo and probabilistic programs~\citep{cusumano2018incremental}, has already been implemented as part of the Gen probabilistic programming system.

Improving the performance of automated involutive MCMC and of flexible probabilistic programming systems like Gen more generally is an important area for future work.
The approach described in this paper is largely dynamic and is not performance-competitive with optimized hand-coded implementations in performance-oriented languages like C.
While our approach is already valuable for use cases where the best performance is not necessary or the expertise or time needed for an optimized hand-coded implementation is not available,
more research into compilers and automatic code generation of custom inference algorithms from high-level user specifications
would broaden the applicability of systems like Gen.

\subsubsection*{Acknowledgements}
This research was supported in part by the US Department of Defense through the the National Defense Science \& Engineering Graduate Fellowship (NDSEG) Program, the DARPA SD2 program (contract FA8750-17-C-0239), the DARPA Machine Common Sense (MCS) program, the DARPA Synergistic Discovery and Design (SD2) program, support from the Intel Corporation, and a philanthropic gift from the Aphorism Foundation.
The authors would also like to thank Feras Saad and Cameron Freer for helpful discussions.

\bibliography{references} 

\begin{thebibliography}{}

\bibitem[Bezanson et~al., 2017]{bezanson2017julia}
Bezanson, J., Edelman, A., Karpinski, S., and Shah, V.~B. (2017).
\newblock Julia: A fresh approach to numerical computing.
\newblock {\em SIAM Review}, 59(1):65--98.

\bibitem[Bingham et~al., 2019]{bingham2019pyro}
Bingham, E., Chen, J.~P., Jankowiak, M., Obermeyer, F., Pradhan, N.,
  Karaletsos, T., Singh, R., Szerlip, P., Horsfall, P., and Goodman, N.~D.
  (2019).
\newblock {Pyro: Deep universal probabilistic programming}.
\newblock {\em The Journal of Machine Learning Research}, 20(1):973--978.

\bibitem[Carpenter et~al., 2017]{carpenter2017stan}
Carpenter, B., Gelman, A., Hoffman, M.~D., Lee, D., Goodrich, B., Betancourt,
  M., Brubaker, M., Guo, J., Li, P., and Riddell, A. (2017).
\newblock {Stan: A probabilistic programming language}.
\newblock {\em Journal of statistical software}, 76(1).

\bibitem[Chang and Pollard, 1997]{chang1997conditioning}
Chang, J.~T. and Pollard, D. (1997).
\newblock Conditioning as disintegration.
\newblock {\em Statistica Neerlandica}, 51(3):287--317.

\bibitem[Cusumano-Towner, 2018]{gen-inv-mcmc}
Cusumano-Towner, M. (2018).
\newblock {Inference library of the Gen probabilistic programming system}.
\newblock
  \url{https://github.com/probcomp/Gen.jl/blob/b9d72b/src/inference/mh.jl#L73-L108}.
\newblock Accessed: 2018-12-27.

\bibitem[Cusumano-Towner et~al., 2018]{cusumano2018incremental}
Cusumano-Towner, M., Bichsel, B., Gehr, T., Vechev, M., and Mansinghka, V.~K.
  (2018).
\newblock Incremental inference for probabilistic programs.
\newblock In {\em Proceedings of the 39th ACM SIGPLAN Conference on Programming
  Language Design and Implementation}, PLDI 2018, pages 571--585. ACM.

\bibitem[Cusumano-Towner et~al., 2019]{cusumano2019gen}
Cusumano-Towner, M.~F., Saad, F.~A., Lew, A.~K., and Mansinghka, V.~K. (2019).
\newblock Gen: A general-purpose probabilistic programming system with
  programmable inference.
\newblock In {\em Proceedings of the 40th ACM SIGPLAN Conference on Programming
  Language Design and Implementation}, pages 221--236. ACM.

\bibitem[Del~Moral et~al., 2006]{del2006sequential}
Del~Moral, P., Doucet, A., and Jasra, A. (2006).
\newblock {Sequential Monte Carlo samplers}.
\newblock {\em Journal of the Royal Statistical Society: Series B (Statistical
  Methodology)}, 68(3):411--436.

\bibitem[Duane et~al., 1987]{hmc}
Duane, S., Kennedy, A.~D., Pendleton, B.~J., and Roweth, D. (1987).
\newblock {Hybrid Monte Carlo}.
\newblock {\em Physics letters B}, 195(2):216--222.

\bibitem[Gehr et~al., 2016]{gehr2016psi}
Gehr, T., Misailovic, S., and Vechev, M. (2016).
\newblock Psi: Exact symbolic inference for probabilistic programs.
\newblock In {\em International Conference on Computer Aided Verification},
  pages 62--83. Springer.

\bibitem[Geiger et~al., 2011]{geiger2011generative}
Geiger, A., Lauer, M., and Urtasun, R. (2011).
\newblock {A generative model for 3D urban scene understanding from movable
  platforms}.
\newblock In {\em CVPR 2011}, pages 1945--1952. IEEE.

\bibitem[Gilks et~al., 1994]{gilks1994language}
Gilks, W.~R., Thomas, A., and Spiegelhalter, D.~J. (1994).
\newblock {A language and program for complex Bayesian modelling}.
\newblock {\em Journal of the Royal Statistical Society: Series D (The
  Statistician)}, 43(1):169--177.

\bibitem[Goodman et~al., 2008]{goodman2012church}
Goodman, N., Mansinghka, V., Roy, D.~M., Bonawitz, K., and Tenenbaum, J.~B.
  (2008).
\newblock Church: a language for generative models.
\newblock In {\em Proceedings of the 24th Annual Conference on Uncertainty in
  Artificial Intelligence}, UAI 2008, pages 220--229. AUAI Press.

\bibitem[Green, 1995]{green1995reversible}
Green, P.~J. (1995).
\newblock {Reversible jump Markov chain Monte Carlo computation and Bayesian
  model determination}.
\newblock {\em Biometrika}, 82(4):711--732.

\bibitem[Grosse et~al., 2012]{grosse2012exploiting}
Grosse, R.~B., Salakhutdinov, R., Freeman, W.~T., and Tenenbaum, J.~B. (2012).
\newblock Exploiting compositionality to explore a large space of model
  structures.
\newblock In {\em Proceedings of the 28th Conference on Uncertainty in
  Artificial Intelligence}, UAI 2012, pages 306--315. AUAI Press.

\bibitem[Hastie and Green, 2012]{hastie2012model}
Hastie, D.~I. and Green, P.~J. (2012).
\newblock {Model choice using reversible jump Markov chain Monte Carlo}.
\newblock {\em Statistica Neerlandica}, 66(3):309--338.

\bibitem[Huelsenbeck et~al., 2004]{huelsenbeck2004bayesian}
Huelsenbeck, J.~P., Larget, B., and Alfaro, M.~E. (2004).
\newblock {Bayesian phylogenetic model selection using reversible jump Markov
  chain Monte Carlo}.
\newblock {\em Molecular biology and evolution}, 21(6):1123--1133.

\bibitem[Lee et~al., 2019]{verified_svi}
Lee, W., Yu, H., Rival, X., and Yang, H. (2019).
\newblock Towards verified stochastic variational inference for probabilistic
  programs.
\newblock {\em Proceedings of the ACM on Programming Languages}, 4(POPL):1--33.

\bibitem[Lew et~al., 2019]{popltraces}
Lew, A.~K., Cusumano-Towner, M.~F., Sherman, B., Carbin, M., and Mansinghka,
  V.~K. (2019).
\newblock Trace types and denotational semantics for sound programmable
  inference in probabilistic languages.
\newblock {\em Proceedings of the ACM on Programming Languages}, 4(POPL):1--32.

\bibitem[Mansinghka et~al., 2018]{mansinghka2018probabilistic}
Mansinghka, V.~K., Schaechtle, U., Handa, S., Radul, A., Chen, Y., and Rinard,
  M. (2018).
\newblock Probabilistic programming with programmable inference.
\newblock In {\em Proceedings of the 39th ACM SIGPLAN Conference on Programming
  Language Design and Implementation}, pages 603--616.

\bibitem[Merrell and Gitter, 2020]{merrell2020inferring}
Merrell, D. and Gitter, A. (2020).
\newblock Inferring signaling pathways with probabilistic programming.
\newblock {\em Proceedings of the Nineteenth European Conference of
  Computational Biology}.

\bibitem[Milch et~al., 2005]{milch2007}
Milch, B., Marthi, B., Russell, S., Sontag, D., Ong, D.~L., and Kolobov, A.
  (2005).
\newblock {BLOG}: Probabilistic models with unknown objects.
\newblock In {\em Proceedings of the Nineteenth International Joint Conference
  on Artificial Intelligence}, IJCAI 2005, pages 1352--1359. Morgan Kaufmann
  Publishers Inc.

\bibitem[Narayanan and Shan, 2020]{narayanan2020symbolic}
Narayanan, P. and Shan, C.-c. (2020).
\newblock Symbolic disintegration with a variety of base measures.
\newblock {\em ACM Transactions on Programming Languages and Systems (TOPLAS)},
  42(2):1--60.

\bibitem[Neklyudov et~al., 2020]{imcmc}
Neklyudov, K., Welling, M., Egorov, E., and Vetrov, D. (2020).
\newblock {Involutive MCMC: A Unifying Framework}.
\newblock {\em arXiv preprint arXiv:2006.16653}.

\bibitem[Pfeffer, 2007]{pfeffer200714}
Pfeffer, A. (2007).
\newblock {The design and implementation of IBAL: A general-purpose
  probabilistic language}.
\newblock {\em Introduction to statistical relational learning}, page 399.

\bibitem[Richardson and Green, 1997]{richardson1997bayesian}
Richardson, S. and Green, P.~J. (1997).
\newblock {On Bayesian analysis of mixtures with an unknown number of
  components (with discussion)}.
\newblock {\em Journal of the Royal Statistical Society: series B (statistical
  methodology)}, 59(4):731--792.

\bibitem[Ritchie et~al., 2016]{ritchie2016deep}
Ritchie, D., Horsfall, P., and Goodman, N.~D. (2016).
\newblock Deep amortized inference for probabilistic programs.
\newblock {\em arXiv preprint arXiv:1610.05735}.

\bibitem[Roberts et~al., 2019]{roberts2019}
Roberts, D.~A., Gallagher, M., and Taimre, T. (2019).
\newblock Reversible jump probabilistic programming.
\newblock In Chaudhuri, K. and Sugiyama, M., editors, {\em Proceedings of
  Machine Learning Research}, volume~89 of {\em Proceedings of Machine Learning
  Research}, pages 634--643. PMLR.

\bibitem[Saad et~al., 2019]{saad2019bayesian}
Saad, F.~A., Cusumano-Towner, M.~F., Schaechtle, U., Rinard, M.~C., and
  Mansinghka, V.~K. (2019).
\newblock Bayesian synthesis of probabilistic programs for automatic data
  modeling.
\newblock {\em Proceedings of the ACM on Programming Languages}, 3(POPL):1--32.

\bibitem[Schaechtle et~al., 2016]{schaechtle2016time}
Schaechtle, U., Saad, F., Radul, A., and Mansinghka, V. (2016).
\newblock Time series structure discovery via probabilistic program synthesis.
\newblock {\em arXiv preprint arXiv:1611.07051}.

\bibitem[Spanbauer et~al., 2020]{spanbauer2020deep}
Spanbauer, S., Freer, C., and Mansinghka, V. (2020).
\newblock Deep involutive generative models for neural {MCMC}.
\newblock {\em arXiv preprint arXiv:2006.15167}.

\bibitem[Tenenbaum et~al., 2011]{tenenbaum2011grow}
Tenenbaum, J.~B., Kemp, C., Griffiths, T.~L., and Goodman, N.~D. (2011).
\newblock How to grow a mind: Statistics, structure, and abstraction.
\newblock {\em Science}, 331(6022):1279--1285.

\bibitem[Tierney, 1998]{tierney1998note}
Tierney, L. (1998).
\newblock {A note on Metropolis-Hastings kernels for general state spaces}.
\newblock {\em Annals of applied probability}, pages 1--9.

\bibitem[Zhi-Xuan et~al., 2020]{zhi2020online}
Zhi-Xuan, T., Mann, J.~L., Silver, T., Tenenbaum, J.~B., and Mansinghka, V.~K.
  (2020).
\newblock {Online Bayesian goal inference for boundedly-rational planning
  agents}.
\newblock {\em arXiv preprint arXiv:2006.07532}.

\end{thebibliography}

\clearpage
\onecolumn
\appendix
\section{APPENDIX}

%

\subsection{Derivation of the pushforward Radon-Nikodym derivative for a special case} \label{sec:radon-nikodym-special-case}
Implementing Algorithm~\ref{alg:involutive-mcmc} requires computing the Radon-Nikodym derivative $d (\mu \circ f^{-1}) / d \mu$.
This section derives that function for the special case in which the involution $f$ can be factored into an involution on a countable set $I$ and a family of bijections on $\mathbb{R}^{n_i}$ for some $n_i$ for each $i \in I$.
Suppose $Z = \{(i, \mathbf{x}) : i \in I, \mathbf{x} \in \mathbb{R}^{n_i}\}$.
Suppose $f_1$ is an involution on $I$ and $n_i = n_{f_1(i)}$ and $f_2$ is a family of continuously differentiable bijections indexed by $i \in I$, such that $f_{2,i} : \mathbb{R}^{n_i} \to \mathbb{R}^{n_i}$.
Also suppose that $f_{2,i} = f_{2,f_1(i)}^{-1}$ for all $i \in I$. That is,
\begin{equation}
f_{2,i}(f_{2,f_1(i)}(\mathbf{x})) = \mathbf{x} \mbox{ for all } \mathbf{x} \in \mathbb{R}^{n_i} \mbox{ and all } i \in I
\end{equation} 
Then, $f : Z \to Z : (i, \mathbf{x}) \mapsto (f_1(i), f_{2,i}(\mathbf{x}))$ is an involution because:
\begin{equation}
f(f(i, \mathbf{x})) = f(f_1(i), f_{2,i}(\mathbf{x})) = (f_1(f_1(i)), f_{2,f_1(i)}(f_{2,i}(\mathbf{x}))) = (i, \mathbf{x})
\end{equation}
Let $\Sigma_n$ and $\mu_n$ denote the Lebesgue $\sigma$-algebra and Lebesgue measure on $\mathbb{R}^n$, respectively.
Let $\Sigma \subset \mathcal{P}(Z)$ be the $\sigma$-algebra of sets of the form $\cup_{i \in I} \{(i, \mathbf{x}) : \mathbf{x} \in K_i\}$ for some $K_i \in \Sigma_{n_i}$ for each $i \in I$.
Let $\mu$ denote the measure on measurable space $(Z, \Sigma)$ given by:
\begin{equation}
\mu(\cup_{i \in I} \{(i, \mathbf{x}) : \mathbf{x} \in K_i\}) := \sum_{i \in I} \mu_{n_i}(K_i)
\end{equation}
We wish to show that the Radon-Nikodym derivative of the pushforward of $\mu$ by $f$ with respect to $\mu$, evaluated at $(i, \mathbf{x})$, is the absolute value of the Jacobian (determinant) of the function $f_{2,i}$ evaluated at $\mathbf{x}$, which is denoted $(J f_{2,i})(\mathbf{x})$:
\begin{equation}
\frac{d (\mu \circ f^{-1})}{d \mu}(i, \mathbf{x})
= \left| (Jf_{2,i})(\mathbf{x})\right|
\end{equation}
Consider $(\mu \circ f^{-1})(A)$ for $K \in \Sigma$:
\begin{align}
(\mu \circ f^{-1})(\cup_{i \in I} \{ (i, \mathbf{x}) : \mathbf{x} \in K_i \})
&= \mu(f^{-1}(\cup_{i \in I} \{ (i, \mathbf{x}) : \mathbf{x} \in K_i \}))\\
&= \mu(\cup_{i \in I} f^{-1}(\{ (i, \mathbf{x}) : \mathbf{x} \in K_i \}))\\
&= \mu(\cup_{i \in I} \{ (f_1(i), \mathbf{x}) : \mathbf{x} \in f_{2,i}(K_i) \})\\
&= \sum_{i \in I} \mu_{n_i}(f_{2,i}(K_i))
\end{align}
It suffices to show that for all $K \in \Sigma$:
\begin{equation}
\int_K \left| (Jf_{2,i})(\mathbf{x})\right| \mu(dz) = \sum_{i \in I} \mu_{n_i}(f_{2,i}(K_i))
\end{equation}
Expanding the left-hand side:
\begin{align}
\int_K \left| (Jf_{2,i})(\mathbf{x})\right| \mu(dz)
= \sum_{i \in I} \int_{K_i} \left| (Jf_{2,i})(\mathbf{x})\right| \mu_{n_i}(d \mathbf{x}) 
= \sum_{i \in I} \mu_{n_i}(f_{2,i}(K_i)) 
\end{align}
where the final step uses Bogachev Theorem 3.7.1 with $g := 1$, and $F := f_{2,i}$.

\subsection{Proof of detailed balance for involution} \label{sec:involution-detailed-balance}
The involutive MCMC kernel is composed of two parts:
An extension of the state space, and an involution on the extended state space.
First, we show detailed balance for the deterministic involution move applied to the extended state space.

%
%

\citet{tierney1998note} gives a class of MCMC kernels based on involutions that satisfy detailed balance.
We now reproduce the result in our notation:
\begin{lemma}[Detailed balance for involution move~\citep{tierney1998note}] \label{lemma:tierney-involution}
Let $(Z, \Sigma, \pi)$ denote a measure space.
Suppose $f$ is a one-to-one function from $Z$ onto $Z$ such that $f^{-1} = f$.
Consider the probability kernel $k$ defined by $k_z(A) := [f(z) \in A] \alpha(z, f(z)) + [z \in A] (1 - \alpha(z, f(z)))$ (where $\alpha(z, f(z))$ gives the probability of accepting a proposed transition from $z$ to $f(z)$).
Let $\nu(dz) := \pi(dz) + (\pi \circ f^{-1})(dz)$.
Let $h(z)$ be a density for $\pi$ with respect to $\nu$.
Let $K := \{z \in Z : h(z) > 0 \mbox{ and } h(f(z)) > 0\}$.
$k$ satisfies detailed balance with respect to $\pi$ if and only if:
\begin{enumerate}
\item $\alpha(z, f(z)) = 0$ for $\pi$-almost all $z \not \in K$
\item $\alpha(z, f(z)) \frac{h(z)}{h(f(z))} = \alpha(f(z), z)$
\end{enumerate}
\end{lemma}

Now we apply Lemma~\ref{lemma:tierney-involution} to our setting where $\pi$ is $\sigma$-finite, there exists a $\sigma$-finite reference measure $\mu$ for measurable space $(Z, \Sigma)$ such that $\pi$ is mutually absolutely continuous with respect to $\mu$, and where the pushforward of $\mu$ by $f$, denoted $\mu \circ f^{-1}$, is absolutely continuous with respect to $\mu$.

In our setting, $\alpha$ is defined as:
\begin{equation}
\alpha(z, f(z)) := \mbox{min}\left\{1, \frac{\frac{d \pi}{d \mu}(f(z))}{\frac{d \pi}{d \mu}(z)} \cdot \frac{d (\mu \circ f^{-1})}{d \mu}(z)\right\}
\end{equation}
This definition of $\alpha$ satisfies:
\begin{equation}
\alpha(z, f(z)) \frac{\frac{d \pi}{d \mu}(z)}{\frac{d \pi}{d \mu}(f(z))} \cdot \left( \frac{d (\mu \circ f^{-1})}{d \mu} (z)\right)^{-1}= \alpha(f(z), z)
\end{equation}
Therefore, to apply Lemma~\ref{lemma:tierney-involution}, it suffices to show $\pi$ has density with respect to $\nu$ (denoted $h(z)$) such that:
\begin{equation}
\frac{h(z)}{h(f(z))} = \frac{\frac{d \pi}{d \mu}(z)}{\frac{d \pi}{d \mu}(f(z))} \left( \frac{d (\mu \circ f^{-1})}{d \mu} (z)\right)^{-1}
\end{equation}
Since $\pi$ and $\pi \circ f^{-1}$ are both absolutely continuous with respect to $\mu$, $\nu$ is also absolutely continuous with respect to $\mu$, and has density:
\begin{equation}
\frac{d \nu}{d \mu}(z) = \frac{d \pi}{d \mu}(z) + \frac{d (\pi \circ f^{-1})}{d \mu}(z) 
\end{equation}
$\pi$ is absolutely continuous with respect to $\nu$, and therefore:
\begin{equation}
\frac{d \pi}{d \mu}(z) = \frac{d \pi}{d \nu}(z) \cdot \frac{d \nu}{d \mu}(z)
\end{equation}
Because $(d \pi / d \mu)(z) > 0$ for all $z \in Z$, $(d \nu / d \mu)(z) > 0$ for all $z \in Z$.
Therefore,
\begin{equation}
h(z)
:= \frac{d \pi}{d \nu}(z)
= \frac{\frac{d \pi}{d \mu}(z)}{\frac{d \nu}{d \mu}(z)}
\end{equation}
Therefore:
\begin{equation}
\frac{h(z)}{h(f(z))} = \frac{\frac{d \pi}{d \mu}(z)}{\frac{d \pi}{d \mu}(f(z))} \cdot \frac{\frac{d \nu}{d \mu}(f(z))}{\frac{d \nu}{d \mu}(z)}
\end{equation}
It suffices to show that:
\[
\frac{\frac{d \nu}{d \mu}(f(z))}{\frac{d \nu}{d \mu}(z)} = \left(\frac{d (\mu \circ f^{-1})}{d \mu}(z)\right)^{-1}
\]

First, we prove a Lemma:
\begin{lemma} \label{lemma:reference-measure-pushforward}
If $(Z, \Sigma)$ is a measurable space and $f : Z \to Z$ is a measurable function that is an involution, $\nu$ and $\mu$ are $\sigma$-finite measures such that $\nu$ is absolutely continuous with respect to $\mu$, and such that the pushforward measures $\nu \circ f^{-1}$ and $\mu \circ f^{-1}$ are both $\sigma$-finite, then
$\nu \circ f^{-1}$ is absolutely continuous with respect to $\mu \circ f^{-1}$ and
\begin{equation}
\frac{d (\nu \circ f^{-1})}{d (\mu \circ f^{-1})}(z) = \frac{d \nu}{d \mu}(f(z))
\end{equation}
\end{lemma}
\begin{proof}
First, $\nu \circ f^{-1}$ is absolutely continuous with respect to $\mu \circ f^{-1}$ because
$(\mu \circ f^{-1})(A) = 0$ implies $\mu(f^{-1}(A)) = 0$ implies $\nu(f^{-1}(A)) = 0$ implies $(\nu \circ f^{-1})(A) = 0$.
To show that $z \mapsto (d \nu / d \mu)(f(z))$ is the Radon-Nikodym derivative $d (\nu \circ f^{-1}) / d (\mu \circ f^{-1})$, it suffices to show that for all $K \in \Sigma$:
\begin{equation}
\int_K \left( \frac{d \nu}{d \mu} (f(z)) \right) ( \mu \circ f^{-1} )(dz) = (\nu \circ f^{-1})(A) := \nu(f^{-1}(A))
\end{equation}
Applying Theorem 3.6.1 in Bogachev with $Y := K$, $y := z$, $x := z'$, $X := f^{-1}(A)$, and $g(y) := (d\nu / d\mu)(f(y))$:
\begin{align}
\int_K \left( \frac{d \nu}{d \mu} (f(z)) \right) ( \mu \circ f^{-1} )(dz)
&= \int_Y g(y) ( \mu \circ f^{-1} )(dy)\\
&= \int_X g(f(x)) \mu(dx) \;\;\;\; [\mbox{Bogachev Theorem 3.6.1}]\\
&= \int_{f^{-1}(A)} (d\nu / d\mu)(f(f(z'))) \mu(dz')\\
&= \int_{f^{-1}(A)} (d\nu / d\mu)(z') \mu(dz')\\
&= \int_{f^{-1}(A)} \nu(dz')
= \nu(f^{-1}(A))
\end{align}

\end{proof}

Now, note that $\nu$ and $\nu \circ f^{-1}$ are the same measure:
\begin{align}
\nu(A) &= \pi(A) + \pi'(A) = \pi(A) + \pi(f^{-1}(A))\\
(\nu \circ f^{-1})(A) &= \pi(f^{-1}(A)) + \pi(f^{-1}(f^{-1}(A))) = \pi(f^{-1}(A)) + \pi(A)
\end{align}
Therefore,
\[
\frac{d (\nu \circ f^{-1})}{d \nu}(z) = 1 \mbox{ for all } z
\]
Expanding $d(\nu \circ f^{-1}) / d \nu$ using the chain rule:
\begin{align}
1 &= \frac{d (\nu \circ f^{-1})}{d \nu}(z)\\
&=   \frac{d (\nu \circ f^{-1})}{d (\mu \circ f^{-1})}(z) \cdot
    \frac{d (\mu \circ f^{-1})}{d \mu}(z) \cdot
    \frac{d \mu}{d \nu}(z) \label{eq:mu-is-absolutely-continuous-with-respect-to-nu}\\
&= \frac{d \nu}{d \mu}(f(z)) \cdot \frac{d (\mu \circ f^{-1})}{d \mu}(z) \cdot \frac{d \mu}{d \mu}(z) \;\;\;\; [\mbox{Lemma~\ref{lemma:reference-measure-pushforward}}]\\
&= \frac{d \nu}{d \mu}(f(z)) \cdot \frac{d (\mu \circ f^{-1})}{d \mu}(z) \cdot \frac{1}{\frac{d \nu}{d \mu}(z)}\\
\left(\frac{d (\mu \circ f^{-1})}{d \mu}(z)\right)^{-1} &= \frac{\frac{d \nu}{d \mu}(f(z))}{\frac{d \nu}{d \mu}(z)}
\end{align}

\subsection{Proof of stationarity for involution} \label{sec:involution-is-stationary}
Detailed balance of the the involution kernel with respect to the measure induced by $\pi$ implies:
\begin{equation}
\int_B \pi(z) k'_z(A) \mu(dz) = \int_K \pi(z) k'_z(B) \mu(dz) \;\; \mbox{ for all } K, B \in \Sigma
\end{equation}
Stationarity with respect to $\pi$ follows by substituting $Z$ for $B$:
\begin{align}
\int_Z \pi(z) k'_z(A) \mu(dz) &= \int_K \pi(z) k'_z(Z) \mu(dz) = \int_K \pi(z) \mu(dz) \;\; \mbox{ for all } A \in \Sigma\\
\end{align}

\subsection{Proof of stationarity for end-to-end kernel} \label{sec:involutive-mcmc-is-stationary}
We are given that the involution is stationary with respect to (the measure induced by) $\pi(x, u) := p(x) q_x(u)$:
\begin{equation}
\int_Z k'_z(A) \pi(z) \mu(dz) = \int_K \pi(z) \mu(dz) \mbox{ for all } A \in \Sigma
\end{equation}
The end-to-end kernel is defined fir all $x \in X_P$ such that $p(x) > 0$ as:
\begin{equation}
k_x(A) := \int_U k'_{x,u}(A \times U) q_x(u) \mu_U(du) \;\; \mbox{for all} \;\; A \in \Sigma_P, x \in X
\end{equation}
Stationarity of the end-to-end kernel with respect to the measure induced by $p$ is:
\begin{equation}
\int_X k_x(A) p(x) \mu_P(dx) = \int_K p(x) \mu_P(dx) \mbox{ for all } A \in \Sigma_P
\end{equation}
Expanding:
\begin{align*}
    \int_X k_x(A) p(x) \mu_P(dx) &= \int_X \left( \int_U k'_{x,u}(A \times U) q_x(u) \mu_Q(du) \right) p(x) \mu_P(dx)\\
    &= \int_{X \times U} k'_{x,u}(A \times U) q_x(u) p(x) (\mu_P \times \mu_Q)(dz)\\
    &= \int_{Z} k'_{x,u}(\{(x', u') \in Z : x' \in A\}) \pi(z) \mu(dz)\\ 
    &= \int_{\{(x, u) \in Z : x \in A\}} \pi(z) \mu(dz)\;\;\; [\mbox{ Stationarity of $k'$ with respect to $\pi$}]\\
    &= \int_{\{(x, u) \in Z : x \in A\}} q_x(u) p(x) \mu(dz)\\
    &= \int_K \left( \int_U q_x(u) \mu_Q(du) \right) p(x) \mu_P(dx)\\
    &= \int_K p(x) \mu_P(dx)
\end{align*}

\subsection{A Sufficient Condition for Involutive MCMC with Dictionaries} \label{sec:technical-requirement}

Our formulation of involutive MCMC requires the following technical condition to hold:
$Z := \{(x, y) \in X \times Y : \pi(x, y) > 0\}$ is $\mu_P \times \mu_Q$-measurable.
We now give a sufficient condition for this to hold, when $X$ and $Y$ are spaces of dictionaries.
Let $D \subseteq \mathcal{K}$ denote the subset of addresses that are discrete (i.e. where $V_k$ is a countable set and $\mu_k$ is the counting measure).
For $\mathbf{x} \in \times_{k \in A} V_k$ let $\mathbf{x} = (\mathbf{d}, \mathbf{c})$ where $\mathbf{d} \in \times_{k \in A \cap D} V_k$ and $\mathbf{c} = \times_{k \in A \setminus D} V_k$, so that $\mathbf{d}$ is the discrete part of $\mathbf{x}$ and $\mathbf{c}$ is the non-discrete part.
\begin{lemma}
Suppose that $p$ and $q$ are such that $p(K, (\mathbf{d}, \mathbf{c})) > 0$ implies
$p(K, (\mathbf{d}, \mathbf{c}')) > 0$ for all $\mathbf{c}' \in \prod_{k \in A \setminus D} V_k$,
and that $q_x(K, (\mathbf{d}, \mathbf{c})) > 0$ where $x = (\tilde{A}, \tilde{\mathbf{d}}, \tilde{\mathbf{c}})$ implies that
$q_{x'}(K, (\mathbf{d}, \mathbf{c}')) > 0$ for all $\mathbf{c}' \in \prod_{k \in A \setminus D} V_k$ and all $x' = (\tilde{A}, \tilde{\mathbf{d}}, \tilde{\mathbf{c}}')$ where $\tilde{\mathbf{c}}' \in \prod_{k \in \tilde{A} \setminus D} V_k$.
Then, $Z := \{(x, y) \in X \times Y : \pi(x, y) > 0\}$ is $\mu \times \mu$-measurable where $\mu$ is the reference measure on traces.
\end{lemma}
\begin{proof}
For $p$ and $q$ satisfying these conditions,
$Z = \cup_{(K_1,K_2,\mathbf{d}_1,\mathbf{d}_2) \in E} \{((K_1, (\mathbf{d}_1, \mathbf{c}_1)), (K_2, (\mathbf{d}_2, \mathbf{c}_2))) : \mathbf{c}_1 \in \times_{k \in K_1 \setminus D} V_k, \mathbf{c}_2 \in \times_{k \in K_2 \setminus D} V_k\}$ for some countable set $E \subseteq \{(K_1, K_2, \mathbf{d}_1, \mathbf{d}_2) : K_1, K_2 \subseteq \mathcal{K}, |K_1| < \infty, |K_2| < \infty, \mathbf{d}_1 \in \times_{k \in K_1 \cap D} V_k, \mathbf{d}_2 \in \times_{k \in K_2 \cap D} V_k\}$ of address sets and discrete choice values for both programs.
The measure of $Z$ is
$(\mu \times \mu)(Z) = \sum_{(K_1,K_2,\mathbf{d}_1,\mathbf{d}_2) \in E} \prod_{k \in (K_1 \setminus D) \cup (K_2 \setminus D)} \mu_k(V_k)$.
\end{proof}
When $p$ and $q$ are defined via probabilistic programs $\mathcal{P}$ and $\mathcal{Q}$ respectively, 
this requirement means that for both the model probabilistic program $\mathcal{P}$ and the auxiliary probabilistic program $\mathcal{Q}$, the support of a random choice that is not discrete cannot depend on the value of another non-discrete random choice.
Additionally, the support of non-discrete random choices in $\mathcal{Q}$ cannot depend on the value of non-discrete random choices in the input $x$, which is a trace of $\mathcal{P}$.
This requirement defines a notion \emph{well-behavedness} for a probabilistic program ($\mathcal{P}$) and an additional notion of well-behavedness for a pair of probabilistic programs that are sequenced one after the other ($\mathcal{P}$ and $\mathcal{Q}$).

\end{document}